\documentclass{article}

\pdfoutput=1

\usepackage[LGR,T1]{fontenc} 
\usepackage[utf8]{inputenc} %
\usepackage{substitutefont} 
\usepackage[polutonikogreek,english]{babel}
\usepackage[largesc]{newtxtext} %
\usepackage[varqu,varl]{zi4}
\usepackage{cabin}
\usepackage[vvarbb,upint,slantedGreek,smallerops]{newtxmath}
\substitutefont{LGR}{\rmdefault}{Tempora} 

\usepackage{textcomp}

\usepackage[margin=1.5in]{geometry}
\usepackage{hyperref}
\usepackage[hyphens]{xurl}
\usepackage{amsmath}
\usepackage{amsfonts}

\usepackage[pdftex]{graphicx}
\usepackage{lastpage}
\usepackage{slashed}

\usepackage{xcolor}

\usepackage{xfrac}

\usepackage{comment}
\usepackage[yyyymmdd,hhmmss]{datetime} 

\usepackage{lastpage}
\usepackage[running,mathlines]{lineno}
\usepackage{isomath}
\usepackage{siunitx}
\usepackage{upgreek}

\usepackage{marginnote}

\usepackage{amsthm}

\newtheorem{theorem}{Theorem}

\newtheorem{lemma}{Lemma}

\newtheorem{postulate}{Postulate}

\newtheorem*{law*}{Law}

\newtheorem{proposition}{Proposition}

\theoremstyle{definition}
\newtheorem{definition}{Definition}

\usepackage{verbatim}

\usepackage{xspace}
\usepackage{braket}
\usepackage{mathtools}
\mathtoolsset{showonlyrefs=true}

\usepackage{siunitx}
\usepackage{physics}

\newcommand{\pr}{\ensuremath{\mathrm{P}}}

\newcommand{\piu}%
{\textrm{\greektext p}}
\newcommand{\eu}%
{\ensuremath{\mathrm{e}}}
\newcommand{\iu}%
{\ensuremath{\mathrm{i}}}

\providecommand{\newoperator}[3]{%
\newcommand*{#1}{\mathop{#2}#3}}

\newcommand{\tran}%
{\textsf{T}}
\newcommand{\herm}%
{\textsf{H}}
\newcommand{\deltau}%
{\textrm{\greektext d}}
\newcommand{\Deltau}%
{\textrm{\greektext D}}

\makeatletter
\providecommand*{\diff}%
{\@ifnextchar^{\DIfF}{\DIfF^{}}}
\def\DIfF^#1{%
\mathop{\mathrm{\mathstrut d}}
\nolimits^{#1}\gobblespace}
\def\gobblespace{%
\futurelet\diffarg\opspace}
\def\opspace{%
\let\DiffSpace\!%
\ifx\diffarg(%
\let\DiffSpace\relax
\else
\ifx\diffarg[%
\let\DiffSpace\relax
\else
\ifx\diffarg\{%
\let\DiffSpace\relax
\fi\fi\fi\DiffSpace}

\newoperator{\loglike}%
{\mathrm{loglike}}{\nolimits}

\newoperator{\notloglike}%
{\mathrm{notloglike}}{\limits}

\makeatother

\newcommand{\Schroedinger}{Schr\"{o}dinger\xspace}

\newcommand{\normalx}[3]{\ensuremath{\mathscr{N}\left(#1 \mid #2,#3\right)}\xspace}

 \newcommand{\bos}{\ensuremath{\mathrm{b}}\xspace}
 \newcommand{\fer}{\ensuremath{\mathrm{f}}\xspace}

\renewcommand{\piu}{\uppi}
\renewcommand{\deltau}{\updelta}

\let\originalpartial\partial
\let\partial\relax
\newrobustcmd*{\partial}{\text{\rotatebox[origin=t]{10}{\scalebox{0.95}[1]{\ensuremath{\originalpartial}}}}\hspace{-0.05em}}

\usepackage{graphicx}
\usepackage{dcolumn}

\newcommand{\CC}{\ensuremath{C}\xspace}
\newcommand{\electron}{\ensuremath{\text{electron}}\xspace}

\newcommand{\photon}{\ensuremath{\text{photon}}\xspace}

\newcommand{\ee}{\ensuremath{\text{e}}\xspace}

\newcommand{\planckbar}{\ensuremath{\hbar}\xspace}

\newcommand{\setConstant}{\ensuremath{\mathscr{C}}\xspace}
\newcommand{\setDecreasing}{\ensuremath{\mathscr{D}}\xspace}
\newcommand{\setIncreasing}{\ensuremath{\mathscr{I}}\xspace}
\newcommand{\setOscillating}{\ensuremath{\mathscr{O}}\xspace}

\newcommand{\hessian}{\ensuremath{\mathbf{H}}\xspace}

\usepackage[normalem]{ulem}

\usepackage{bm}

\newcommand{\Identitymatrix}{\ensuremath{\mathbf{I}}\xspace}

\newcommand{\Sigmabold}{\ensuremath{\bm{\mathrm{\Sigma}}}\xspace}

\newcommand{\kernBeforeIntegral}{\ensuremath{\kern-0.17em}\xspace}

\newcommand{\physics}{physics\xspace}

\newcommand{\mechanics}{mechanics\xspace}

\newcommand{\eqdefA}{\overset{\Updelta}{=}\xspace}

\newcommand{\entropyS}{\ensuremath{\mathrm{S}\xspace}}

\newcommand{\setEvolutions}{\ensuremath{\mathscr{E}}\xspace}

\newcommand{\QCurve}{QCurve\xspace}
\newcommand{\QCurves}{QCurves\xspace}

\newcommand{\setAll}{\ensuremath{\mathscr{E}}\xspace}
\newcommand{\CPT}{\ensuremath{\mathrm{CPT}}\xspace}
\newcommand{\QFT}{\ensuremath{\mathrm{QFT}}\xspace}

\usepackage[mathscr]{euscript}

\usepackage[en-US]{datetime2}

\begin{document}

\title{
Quantum-Entropy Physics
}
\author{{Davi Geiger} and {Zvi M.\ Kedem}
\\
Courant Institute of Mathematical Sciences\\
New York University}

\maketitle

\begin{abstract}

All the laws of physics are time-reversible. Time arrow emerges only when ensembles of classical particles are treated probabilistically, outside of physics laws, and the entropy and the second law of thermodynamics are introduced. In quantum physics, no mechanism for a time arrow has been proposed despite its intrinsic probabilistic nature. In consequence, one cannot explain why an electron in an excited state will ``spontaneously'' transition into a ground state as a photon is created and emitted, instead of continuing in its reversible unitary evolution. To address such phenomena, we introduce an entropy for quantum physics, which will conduce to the emergence of a time arrow.

The entropy is a measure of randomness over the degrees of freedom of a quantum state. It is dimensionless; it is a relativistic scalar, it is invariant under coordinate transformation of position and momentum that maintain conjugate properties and under CPT transformations; and its minimum is positive due to the uncertainty principle.

To excogitate why some quantum physical processes cannot take place even though they obey conservation laws, we partition the set of all evolutions of an initial state into four blocks, based on whether the entropy is (i) increasing but not a constant, (ii) decreasing but not a constant, (iii) a constant, (iv) oscillating. We propose a law that in quantum physics entropy (weakly) increases over time. Thus, evolutions in the set (ii) are disallowed, and evolutions in set (iv) are barred from completing an oscillation period by instantaneously transitioning  to a new state. This law for quantum physics limits physical scenarios beyond conservation laws, providing causality reasoning by defining a time arrow.

\end{abstract}

\tableofcontents

\section{Introduction}

Today's classical and quantum \physics laws are time-reversible and a time arrow emerges in \physics  only when a probabilistic behavior of ensembles of particles is considered and physical  causes for various phenomena are set. However, for many physical events that do obey the time arrow, such as ``an excited electron in the hydrogen atom  jumps to the ground state while emitting radiation,'' described as ``a spontaneous emission,'' no physical explanations for
their causes are known. While transition probabilities obtained from Fermi's golden rule \cite{dirac1927quantum,fermi1950nuclear} for the hydrogen atoms are highly accurate, this rule, derived from an energy perturbation method, cannot be a source of the time arrow or causality.   More generally, although conservation laws must be obeyed, they do not provide a physical account for the instant particles are created. For example, in an excited hydrogen, when a photon is emitted and the electron jumps to the ground state, we argue, an instantaneous irreversible process occurs where a photon is created.
No physical explanation for the  cause of such phenomena is known,  and if it were known, it could shed light on the time arrow.

Quantum \physics introduces probability  as intrinsic to the description of a single-particle system. A probability $\pr(a)$ is assigned to each value of $a \in \mathscr{A}$, where $\mathscr{A}$ is an observable.  For a finite set $\mathscr{A}=\{a_i :  i = 1, 2, \dots, N\}$, the Shannon entropy, $\text{H}=-\sum_{i=1}^N \pr(a_i) \log_2 \pr(a_i)$, is a measure of information about $ \mathscr{A}$. The more concentrated is the probability around a few values of $\mathscr{A}$, the more information is provided about an observable, and the lower is the entropy.  Entropy is a measure of such information, or of lack of information.

Extending the concept of entropy to continuous variables, continuous distributions, and to quantum \mechanics has been challenging. For example, von Neumann's entropy \cite{von2018mathematical} requires the existence of classical statistics elements (mixed states) in order not to vanish, and consequently it assigns zero entropy to all one-particle systems. Our goal is to assign an entropy measure for a one-particle system that can be extended to multiple particles. Therefore, we cannot consider von Neumann's entropy as a starting point for an entropy measure in the quantum domain.

Another challenge for proposing an entropy is that in quantum \mechanics a one-particle system is described by a quantum state $\ket{\psi_t}$, which is a ray in Hilbert space, while in quantum field theory such a state is described by an operator $\Psi$ acting on the vacuum state $\ket{0}$  and written as a linear combination of a creation and an annihilation operator. We do propose an entropy that is applicable in both scenarios: the  quantum \mechanics and  the quantum field  theory.

We require the entropy  (i) to account for all the degrees of freedom of a state, (ii)  to be a measure of randomness of such a state, (iii) to be  invariant under the applicable continuous and discrete transformations.    In classical physics, Boltzmann entropy and Gibbs entropy, and their respective H-theorems  \cite{gibbs2014elementary}, are formulated in phase space, reflecting the degrees of freedom of a system. In quantum \mechanics position and momentum are conjugate operators and their eigenstates  will allow to describe a phase space,  while in quantum field theory the position is ``demoted'' to a  variable, and the  Fourier transform of the fields introduces a spatial frequency variable that together with position compose a quantum field phase space coordinate system.
In addition to position and momentum  we  must also consider the degree of freedom associated with the spin operator and  as one expands  to the standard model, other degrees of freedom, such as flavor, will need to be incorporated as well. These internal degrees of freedom are captured by representing the states with more complex structure, such as  Dirac spinors for fermions and  the two polarization components for photons, as well as providing the groups of transformations  that such structures follow. Applicable continuous transformations to the states include change of coordinates and special relativity, and applicable discrete transformations include Charge Conjugation (C), Parity(P), Time Reversal (T), and their concatenation (CPT).

We  propose an entropy  that satisfies all the above requirements. We note that in contrast to  the entropy in classical physics, the minimum value of our entropy must be positive due to the standard uncertainty principle~\cite{Robertson1929}.

Having defined an entropy, we then propose an entropy law, inspired by the second law of thermodynamics, stating that only states that evolve with (weakly) increasing entropy are allowed. This law provides the time arrow as the arrow of information loss.
Exploring the consequences of this law, we analyze physical phenomena  reported as spontaneous transitions  or particle transformation that may be caused by the proposed entropy law.

The paper is organized as follows.

In Section~\ref{sec:quantum-entropy-def} we propose an entropy measure of randomness of a quantum state.

In Section~\ref{sec:entropy}  we examine various properties of our proposed entropy. First, in Section~\ref{subsec:minimum-entropy}  we prove its minimum. Then,  we prove its invariant properties under coordinate transformations in phase space in Section~\ref{subsec:coordinate-transformation}, and under CPT transformations in Section~\ref{subsec:cpt}. In Section~\ref{subsec:Lorentz-transformation} we show  that the entropy is a scalar in special relativity.   We conclude Section~\ref{sec:entropy}  with a proposal for a \QCurve structure to analyze particles' evolution by partitioning the space of such \QCurves according to the entropy behavior during the evolution.

In Section~\ref{sec:dispersion-coherent-states} we prove that because of the dispersion property of the fermion Hamiltonian the entropy of  \QCurves of  coherent states increases with time.

In Section~\ref{sec:oscillation} we study scenarios where the entropy oscillates. In Section~\ref{subsec:Fermi} we revisit the foundations of Fermi's golden rule and derive an exact formula for the transition of states when an eigenstate $\psi_0$ of a Hamiltonian $H^0$ evolves under a Hamiltonian $H^0+H^{\mathrm{I}}$. To obtain Fermi's golden rule, we simply rely on  $H^{\mathrm{I}}$ contributing little  compared to $H^0$.  We then show that the entropy oscillates during the evolution of $\psi_0$. In    Section~\ref{subsec:time-reflection} we expand the time reversal transformation by adding time translation, obtaining  time reflection. By combining it with CP transformation, we create scenarios where a decreasing-entropy \QCurve can be transformed into an increasing-entropy \QCurve, while neutral particles transform into their anti-particles. Finally in Section~\ref{subsec:two-particle-system}, we study the collision of two particles, described by coherent states the entropy of each particle alone is increasing, and show that an oscillation may occur due to the entanglement interference between the two particles, when they are close to each other.

We propose in Section~\ref{sec:entropy-law} an entropy law for quantum \mechanics and quantum field theory: the entropy is an increasing function of time. Then we analyze the consequences of this law. In Section~\ref{subsec:stab-hydr-ground} we show that the entropy increases when in a hydrogen atom an electron transitions from an excited state to the ground state while emitting a photon. In Section~\ref{subsec:kaons} we speculate that kaons and neutrinos, due to their entropy-oscillating trajectories, transform into anti-particles so that the entropy can increase. We conclude with a speculation in Section~\ref{subsec:collisions} that a collision of particles creates to new particles emission in order for the entropy to increase, when at the collision the entropy enters an oscillating period.

Section~\ref{sec:conclusion} concludes the paper with some thoughts about future topics of investigation.

\section{Quantum Entropy}
\label{sec:quantum-entropy-def}

We propose a definition of a one-particle entropy, a  measure of randomness of all the degrees of freedom,  to be
\begin{align}
\entropyS&= -\kernBeforeIntegral\int  \, \rho_{\mathrm{r}} (\mathbf{r},t)  \rho_p (\mathbf{p},t) \ln \left ( \rho_{\mathrm{r}} (\mathbf{r},t)  \rho_p (\mathbf{p},t) \, \hbar^3\right) \,
 \diff^3\mathbf{r}\, \diff^3\mathbf{p}\, ,\\
 &= \entropyS_{\mathrm{r}}+ \entropyS_{p}-3\ln \hbar\,,
\label{eq:Relative-entropy}
\end{align}
where by the Born's rule $\rho_{\mathrm{r}}(\mathbf{r},t)=|\psi(\mathbf{r},t)|^2$, $\rho_p(\mathbf{p},t)=|\tilde \phi(\mathbf{p},t)|^2$, and $\entropyS_{\mathrm{r}}\eqdefA  -\kernBeforeIntegral\int  \rho_{\mathrm{r}} (\mathbf{r},t) \ln \rho_{\mathrm{r}} (\mathbf{r},t) \, \diff^3\mathbf{r}$, and analogously for $\entropyS_{\mathrm{p}}$. The reduced Planck constant $\hbar$ is used, so the entropy is dimensionless, and thus we propose a dimensionless phase space  volume element  to be $\frac{1}{\hbar^3}\diff^3\mathbf{r}\, \diff^3\mathbf{p}$,   and a  dimensionless  probability density in phase space to be  $\rho_{\mathrm{r}} (\mathbf{r},t)  \rho_p (\mathbf{p},t)\, \hbar^3$. Thus, the entropy is dimensionless and changing the units of measurements will not change it. It is clear that the proposed  entropy has some of its foundations in the work of  Gibbs~\cite{gibbs2014elementary} and  Jaynes~\cite{jaynes1965gibbs}.

In a quantum \mechanics setting, the degrees of freedom of the conjugate operators $\hat{\mathbf {r}}$ and $\hat{\mathbf {p}}$ are measured through the projection of the state $\ket{\psi_t}$ on their eigenvectors $\ket{\mathbf {r}}$ and $\ket{\mathbf {p}}$, respectively, that is, the state of a particle in phase space is described by  $\left ( \psi(\mathbf{r},t)=\bra{\mathbf{r}}\ket{\psi_t},\, \phi(\mathbf{k},t)=\bra{\mathbf{k}}\ket{\psi_t}\right)$ and the entropy is a function of the state.

In quantum field theory (QFT),  the states in the phase space  are described through  the operators $\Psi(\mathbf{r},t)$ and  $\Phi(\mathbf{k},t)$, where $\mathbf{r}$ and $t$ become parameters describing space-time, and the spatial frequency $\mathbf{k}$ is the Fourier transform of $\mathbf{r}$, and can be interpreted as a momentum variable $\mathbf{p}=\hbar\mathbf{k} $.   These operators act on the vacuum sate $\ket{0}$ (we use the Fock space occupancy representation), that is, a one-particle state in phase space is given by $\left ( \Psi(\mathbf{r},t)\ket{0},\, \Phi(\mathbf{k},t)\ket{0}\right) $, where $\Phi(\mathbf{k},t)$ is the Fourier transform of $\Psi(\mathbf{r},t)$.  To describe the entropy of a single particle in phase space, the density functions used to compute the  entropy are the magnitude square of the quantum fields  states, that is,  $\rho^{\QFT}_{\mathrm{r}}(\mathbf{r},t)=|\Psi(\mathbf{r},t)\ket{0}|^2=\bra{0}\Psi^{\dagger}(\mathbf{r},t)\Psi(\mathbf{r},t)\ket{0}$ and  $\rho^{\QFT}_{\mathrm{k}}(\mathbf{k},t)=|\Phi(\mathbf{k},t)\ket{0}|^2=\bra{0}\Phi^{\dagger}(\mathbf{k},t)\Phi(\mathbf{k},t)\ket{0}$. Note that in this case, $\entropyS_{k}=\entropyS_{p}-3\ln \hbar$, so that, using the variables $\mathbf{r},\mathbf{k}$ already make the phase space units dimensionless.

Note that for fermions the conjugate momentum field  operator becomes $\piu_{\psi}=\iu \Psi^{\dagger}(\mathbf{r},t)$, and so all the degrees of freedom of a fermion are expressed by the  bi-spinor $\Psi(\mathbf{r},t)$. Nevertheless, the entropy is  described above in phase space coordinates  $(\mathbf {r},\mathbf {k})$,  with both densities, $\rho^{\QFT}_{\mathrm{r}}(\mathbf{r},t)$ and $ \rho^{\QFT}_{\mathrm{k}}(\mathbf{k},t)$. In the rest of the paper the superscript $\QFT$ will generally be dropped as it will be clear whether we are using the representation of quantum \mechanics or of QFT.

The degrees of freedom associated with the spin operator are captured by the entropy, since they are represented as `` internal degrees of freedom of the state''  in both quantum \mechanics and QFT. More precisely,   $\psi(\mathbf{r},t)$ or  $\Psi(\mathbf{r},t)$ and $\phi(\mathbf{p},t)$ or $\Phi(\mathbf{p},t)$ are described by  Dirac spinors, and the Gauge fields have the polarization degree of freedom, and the densities used to compute the entropy depend on such internal degrees of freedom.

A natural extension of this  entropy to an $N$-particle system is
 \begin{align}
     S &= - \kernBeforeIntegral\int \frac{\diff^3 \mathbf {r}_1 \diff^3\mathbf {p}_1}{\hbar^3} \hdots \frac{\diff^3 \mathbf {r}_N \diff^3\mathbf {p}_N}{\hbar^3}  \,   \rho_{\mathrm{r}}(\mathbf {r}_1,\hdots, \mathbf {r}_N ,t)  \rho_{\mathrm{k}}(\mathbf {k}_1,\hdots, \mathbf {k}_N ,t)
     \\
     & \qquad \times    \hbar^{3N}\,  \ln \left (  \rho_{\mathrm{r}}(\mathbf {r}_1,\hdots, \mathbf {r}_N ,t) \rho_{\mathrm{k}}(\mathbf {k}_1,\hdots, \mathbf {k}_N ,t) \hbar^{3N}\right)
     \\
     &= -\kernBeforeIntegral\int \diff^3 \mathbf {r}_1 \hdots  \kernBeforeIntegral\int \diff^3 \mathbf {r}_N \,   \rho_{\mathrm{r}}(\mathbf {r}_1,\hdots, \mathbf {r}_N ,t)
     \ln \rho_{\mathrm{r}}(\mathbf {r}_1,\hdots, \mathbf {r}_N ,t)
     \\
     & \quad -  \kernBeforeIntegral\int \diff^3 \mathbf {p}_1 \hdots  \kernBeforeIntegral\int \diff^3 \mathbf {p}_N\,   \rho_{\mathrm{p}}(\mathbf {p}_1,\hdots, \mathbf {p}_N ,t)
     \ln \rho_{\mathrm{p}}(\mathbf {p}_1,\hdots, \mathbf {p}_N ,t)
     \\
     &\quad - 3N  \ln \hbar\, ,
     \label{eq:entropy-many-particles}
 \end{align}
 where $\rho_{\mathrm{r}}(\mathbf {r}_1,\hdots, \mathbf {r}_N ,t)=|\psi(\mathbf {r}_1,\hdots, \mathbf {r}_N ,t)|^2$ and $\rho_{\mathrm{p}}(\mathbf {p}_1,\hdots, \mathbf {p}_N ,t)=|\phi(\mathbf {p}_1,\hdots, \mathbf {p}_N ,t)|^2 $ are defined in quantum \mechanics via the projection of the state $\ket{\psi_t}^N$ of $N$ particles, fermions (f) or bosons (b), on the position $\bra{\mathbf {r}_1}\hdots\bra{\mathbf {r}_N} $ and the momentum $\bra{\mathbf {p}_1}\hdots\bra{\mathbf {p}_N} $ coordinate systems. The state $\ket{\psi_t}^N$ is defined in  Fock spaces,  the product of $N$ Hilbert spaces, and  requiring  the construction of combinatorics to describe  indistinguishable particles. For fermions that is done through the permanent, to guarantee that two fermions with the same spin do not occupy the same set of coordinates. In QFT, $\rho_{\mathrm{r}}(\mathbf {r}_1,\hdots, \mathbf {r}_N ,t)$ and $\rho_{\mathrm{k}}(\mathbf {k}_1,\hdots, \mathbf {k}_N ,t)$  are defined through the same combinatorics to describe indistinguishable particles applied to the quantum fields of each particle and then create $N$ particles from the vacuum. The use of the spatial frequency representation  $\mathbf {k}$ indicates the use of QFT.  For example, for two fermions  the combinatorics yields $\Psi_1(\mathbf{r}_1,\mathbf{r}_2,t)\ket{0}= \frac{1}{\sqrt{2}}\left (\Psi_1(\mathbf{r}_1,t)\Psi_2(\mathbf{r}_2,t)-\Psi_1(\mathbf{r}_2,t)\Psi_2(\mathbf{r}_1,t)\right)\ket{0}$ and $\Phi_1(\mathbf{k}_1,\mathbf{k}_2,t)\ket{0}= \frac{1}{\sqrt{2}}\left (\Phi_1(\mathbf{k}_1,t)\Phi_2(\mathbf{k}_2,t)-\Phi_1(\mathbf{k}_2,t)\Phi_2(\mathbf{k}_1,t)\right)\ket{0}$, where $\Phi_{1,2}(\mathbf{k}_{1,2},t)$  are the Fourier transform of $\Psi_{1,2}(\mathbf{k}_{1,2},t)$. For  bosons, instead of a subtraction of the two terms the addition of the same two terms is applied. In QFT the term $-3N\ln \hbar$ is absorbed by the entropy term of  the spatial frequency, that is, using the variables $\mathbf{r}_1\hdots \mathbf{r}_N,\mathbf{k}_1\hdots \mathbf{k}_N$ already makes the phase space units dimensionless.

\section{Entropy Properties and Behaviors}
\label{sec:entropy}

In this section we  establish the properties of the entropy that support its relevance in analyzing quantum systems. In each section we use either the quantum \mechanics setting or QFT settings, but the results that follow are valid for both settings.  We will clarify in each section which settings we are working with. In Section~\ref{subsec:minimum-entropy} the lowest possible entropy value is established.  We then show that the entropy is invariant under  (i) continuous coordinate transformations  in Section~\ref{subsec:coordinate-transformation}, (ii) discrete CPT transformations, in Section~\ref{subsec:cpt}, (iii) special relativity, in Section~\ref{subsec:Lorentz-transformation}.  In Section~\ref{subsec:entropy-partition}
we construct the  structure of \QCurves to partition the space of  the state evolutions according to the behavior of the entropy of such \QCurves.

\subsection{The Minimum Value of the Entropy}
\label{subsec:minimum-entropy}
The third law of thermodynamics establishes the value of $0$ as the minimum entropy. However, the minimum of the  quantum entropy must be positive because due to the uncertainty principle, particles can not be localized with zero velocity.
\begin{proposition}
 \label{proposition:lower-bound-S}
{The minimum of the entropy is $3 (1+\ln \piu)$.}
\end{proposition}
\begin{proof}
The result follows from the entropy definition~\eqref{eq:Relative-entropy} and from the entropy uncertainty principle $\entropyS_{\mathrm{r}}+ \entropyS_p\ge 3 \ln \ee \piu \hbar$, see \cite{hirschman1957note,beckner1975inequalities,  bialynicki1975uncertainty}.
\end{proof}
This  lower bound of  the entropic uncertainty principle is tighter than the bound of the {standard} uncertainty principle  \cite{Robertson1929}, as shown in \cite{beckner1975inequalities,  bialynicki1975uncertainty}. The mathematical proof of \cite{beckner1975inequalities} does not depend on the physics setting, while the connection to the uncertainty principle by \cite{bialynicki1975uncertainty} uses the quantum \mechanics setting and can be straightforwardly extended to QFT.  Coherent states reach the minimum of the entropy as well as the minimum of the {standard} uncertainty principle.
One may notice that the dimensionless element of volume of integration to define the entropy  will not contain a particle unless  $\frac{1}{\hbar^3}\diff^3\mathbf{r}\, \diff^3\mathbf{p}\ge 1$, due to the  uncertainty principle. One may interpret  this as a necessity of discretizing the phase space. We  note that  the  minimum of the entropy for the discrete sum  is also $3 (1+\ln \piu)$ as shown in~\cite{dembo1991information}.

As the entropy is an additive property, it follows  from \eqref{eq:entropy-many-particles} that for an $N$-particle system the minimum entropy is $3N (1+\ln \piu)$.

\subsection{Entropy Invariance Under Phase-Space  Transformations}
\label{subsec:coordinate-transformation}

We investigate two types of transformations of the phase space. The first is a point transformation of coordinates and the second  is a translation in phase space of a  quantum reference frame \cite{aharonov1984quantum}. This section is described in quantum mechanics setting, and a QFT derivation is  sketched for completion.

Consider a  transformation of position coordinates $F: \mathbf{r} \mapsto \mathbf{r}'$. In quantum \mechanics, such a coordinate transformation must be a point transformation, which induces  the new conjugate momentum operator  (see DeWitt \cite{dewitt1952point}):
\begin{align}
    \label{eq:conjugate-momentum-DeWitt}
    \hat{\mathbf{p}}'=-\iu \hbar \left [\nabla_{\mathrm{r}'} + \frac{1}{2 }J^{-1}(\mathbf{r}')\nabla_{\mathrm{r}'}\cdot J(\mathbf{r}')\right]\, ,
\end{align}
where  $J(\mathbf{r}') = \frac{\partial \mathbf{r}(\mathbf{r}')}{\partial \mathbf{r}'}$ is the Jacobian  of $F^{-1}$.

\begin{proposition}
 \label{proposition:change-coordinates-S}
The entropy  is invariant under a  point transformation of  coordinates.
\end{proposition}
\begin{proof}
We are given an entropy $\entropyS$ in phase space relative to a conjugate Cartesian pair of coordinates $(\mathbf{r}, \mathbf{p})$. We can ignore  the entropy term $-3\ln \hbar$  because it does not vary under point transformations.  Consider a  point transformation  $F: \mathbf{r} \mapsto \mathbf{r}'$ and let $\mathbf{p}'$ denote the momentum conjugate to $\mathbf{r}'$.
As the probabilities in infinitesimal volumes are invariant under point transformations,
\begin{align}
    \label{eq:requirement}
    |\psi'(\mathbf{r}'(\mathbf{r} ))|^2 \diff^3 \mathbf{r}' =|\psi(\mathbf{r} )|^2 \diff^3 \mathbf{r}
  \qquad \!\!\text{and} \!\!\qquad
    |\tilde \phi'(\mathbf{p}'(\mathbf{p} ))|^2 \diff^3 \mathbf{p}'=|\tilde \phi(\mathbf{p} )|^2 \diff^3 \mathbf{p}\,.
\end{align}
Thus, the Born's rule that the probability density fuctions are $|\psi'(\mathbf{r}')|^2 $ and $|\tilde \phi'(\mathbf{p}')|^2$ remains valid.  Under coordinate transformation, the Jacobian satisfies
\begin{align}
\label{eq:position-volume-invariance}
    \det J(\mathbf{r}') \, \diff ^3 \mathbf{r}'=\diff ^3 \mathbf{r}\,.
\end{align}
Combining \eqref{eq:requirement} and \eqref{eq:position-volume-invariance} we have
\begin{align}
\label{eq:psi-prime}
   \bra{\mathbf{r}'}\ket{\psi}=\frac{1}{\sqrt{\det J(\mathbf{r}')}} \psi'(\mathbf{r}') = \psi(\mathbf{r}(\mathbf{r}'))\, ,
\end{align}
so that the inifinitesimal probabilty $|\bra{\mathbf{r}'}\ket{\psi}|^2  \det J(\mathbf{r}') \, \diff ^3 \mathbf{r}' = |\psi'(\mathbf{r}')|^2 \diff^3 \mathbf{r}'$ is an invariant.

Considering the Fourier basis $\bra{\mathbf{p}}\ket{\mathbf{r}}=\frac{1}{(2\pi)^{\frac{3}{2}}} \eu^{-\iu \mathbf{r}\cdot \mathbf{p}}$ combined with \eqref{eq:psi-prime} leads to
\begin{align}
 \tilde \phi(\mathbf{p})& =\bra{\mathbf{p}}\ket{\psi}=\int \det J(\mathbf{r}') \,\diff^3 \mathbf{r}'\, \bra{\mathbf{p}}\ket{\mathbf{r}'}\bra{\mathbf{r}'}\ket{\psi}
    \\
    &=  \frac{1}{(2\piu)^{\frac{3}{2}}}\int \sqrt{\det J(\mathbf{r}')}\,  \psi'(\mathbf{r}') \, \eu^{-\iu \mathbf{r}'\cdot \mathbf{p}}\,   \diff^3 \mathbf{r}' \,.
\end{align}

DeWitt noted that in the momentum space there is a transformation $G: \mathbf{p} \mapsto \mathbf{p}'$, specified by~\eqref{eq:conjugate-momentum-DeWitt}  up to an arbitrarily  function $g(\mathbf{p}')=\det J(G^{-1})(\mathbf{p}')$, that is, up to  the  determinant of the Jacobian of $G^{-1}$.  This degree of freedom to specify $g(\mathbf{p}')$ is clearly seen since the volume elements scale according to the determinant of the Jacobian, that is $g(\mathbf{p}')\, \diff ^3 \mathbf{p}'=\diff ^3 \mathbf{p}$, and similarly to~\eqref{eq:psi-prime} we define

\begin{align}
\label{eq:phi-prime}
  \bra{\mathbf{p}'}\ket{\psi} & =\frac{1}{\sqrt{ g(\mathbf{p}')}}\tilde \phi'(\mathbf{p}')=\tilde \phi(\mathbf{p}(\mathbf{p}'))
  \\
  &\Downarrow
  \\
  |\bra{\mathbf{p}'}\ket{\psi}|^2  g(\mathbf{p}')\, \diff ^3 \mathbf{p}' &=|\tilde \phi'(\mathbf{p}')|^2 \diff ^3 \mathbf{p}' \,,
\end{align}
so that we have an infinitesimal probability invariant in momentum space satisfying the Born's rule, that is,  satisfying~\eqref{eq:requirement}. Thus, we can scale $G$ by scaling $\det J(G^{-1})(\mathbf{p}')$ according to any function  $f(\mathbf{p}')$, while also scaling $\bra{\mathbf{p}'}\ket{\psi}$ according to $\frac{1}{\sqrt{f(\mathbf{p}')}}$, and that will satisfy the conjugate properties and be  a valid solution.

Thus
\begin{align}
    \entropyS_{\mathrm{r}}+\entropyS_{\mathrm{p}}   &= -\kernBeforeIntegral\int
 \diff^3\mathbf{r}\, \diff^3\mathbf{p}\, \rho_{\mathrm{r}} (\mathbf{r},t)  \rho_p (\mathbf{p},t) \ln \left ( \rho_{\mathrm{r}} (\mathbf{r},t)  \rho_p (\mathbf{p},t) \right)
 \\
 &=-\kernBeforeIntegral\int
 \diff^3\mathbf{r}'\, \diff^3\mathbf{p}' \,\rho'_{r'} (\mathbf{r}',t)  \rho'_{p'} (\mathbf{p}',t) \ln \left [  \frac{1}{g(\mathbf{p}')\, \det J(\mathbf{r}') }\, \rho'_{r'}(\mathbf{r}',t) \,\rho'_{p'} (\mathbf{p}', t)  \right]
 \\
 &=\entropyS_{\mathrm{r}'}+\entropyS_{\mathrm{p}'}-\kernBeforeIntegral\int
 \diff^3\mathbf{r}'\, \diff^3\mathbf{p}' \,\rho'_{r'} (\mathbf{r}',t)  \rho'_{p'} (\mathbf{p}',t) \ln  \left [  \frac{\det J^{-1}(\mathbf{r}')}{g(\mathbf{p}')\,  }\,  \right]
 \\
 &= \entropyS_{\mathrm{r}'}+\entropyS_{\mathrm{p}'} - \langle \ln \det J^{-1}(\mathbf{r}')\rangle_{\rho'_{r'}} + \langle \ln  g(\mathbf{p}')
 \rangle_{\rho'_{p'}}
 = \entropyS_{\mathrm{r}'}+\entropyS_{\mathrm{p}'} \,,
\end{align}
where
\begin{align}
\langle \ln \det J^{-1}(\mathbf{r}')\rangle_{\rho'_{r'}} &\eqdefA \int  \rho'_{r'} (\mathbf{r}',t)  \ln  \det J^{-1}(\mathbf{r}')\,,\\
\langle\ln  g(\mathbf{p}')\rangle_{\rho'_{p'}} &\eqdefA \int  \rho'_{p'} (\mathbf{p}',t)   \ln    g(\mathbf{p}')\, \diff^3\mathbf{p}' \,.
\end{align}
But  $g(\mathbf{p}')$ must be chosen as to satisfy
\begin{align}
\langle\ln  g(\mathbf{p}')\rangle_{\rho'_{p'}}  = \langle\ln \det J^{-1}(\mathbf{r}')\rangle_{\rho'_{r'}}=\langle\ln \frac{1}{\det J(\mathbf{r}')}\rangle_{\rho'_{r'}}\,.
\end{align}
\end{proof}

Proposition~\ref{proposition:change-coordinates-S} also applies to QFT, where we also consider the spatial coordinates transformations according to a function $F: \mathrm{r} \mapsto \mathrm{r}'$. The QFT operators $\Psi(\mathrm{r},t)$ and $\Phi(\mathrm{k},t)$ are related by the Fourier transform, and $\Psi'(\mathrm{r}',t)$ and $\Phi'(\mathrm{k}',t)$ are also  related by the Fourier transform.  The proposition uses the freedom in the spatial frequency space to choose a scaling of the transformation $G: \mathrm{k} \mapsto \mathrm{k}'$ preserving the invariant infinitesimal probabilities. The proposition elucidates that an invariant entropy  under phase space transformations requires that  $\langle\ln  \det G^{-1}(\mathbf{k}')\rangle_{\rho'_{k'}}  = \langle\ln \det J^{-1}(\mathbf{r}')\rangle_{\rho'_{r'}}$.

We now investigate another symmetry transformation.  When a  quantum reference frame \cite{aharonov1984quantum}  is translated by $x_0$ along
$x$, the state $\ket{\psi_t}$ in position representation becomes  $\psi(x-x_0,t) = \bra{x-x_0}\ket{\psi_t}= \bra{x}\hat T_P(-x_0)\ket{\psi_t}$, where  $\hat T_P(-x_0)=\eu^{\iu x_0 \, \hat{P}}$, and $\hat{P}$ is the momentum operator conjugate to $\hat{X}$.  When the reference frame is translated by $p_0$ along  $p$, the  state $\ket{\psi_t}$  in momentum representation becomes  $\tilde \phi(p-p_0,t)= \bra{p-p_0}\ket{\psi_t}= \bra{p}\hat T_X(-p_0)\ket{\psi_t}$, where  $\hat T_X(-p_0)=\eu^{\iu p_0 \, \hat{X}}$, and $\hat X$ is the position operator conjugate to $\hat{P}$.

 \begin{lemma}
 \label{lemma:frame-invariance}
 Consider a state $\ket{\psi_t}$. The entropy $\entropyS$ is invariant under a change of quantum reference frame by translations along $x$ and along $p$.
 \end{lemma}
 \begin{proof}

 We start by showing that $\entropyS_{x}= -\kernBeforeIntegral\int_{-\infty}^{\infty} \diff x \, |\psi(x,t)|^2 \ln |\psi(x,t)|^2$ is invariant under:
\begin{enumerate}
    \item[(i)]  translations along $x$ by any $x_0$  because
    $|\psi(x,t)|^2$ becomes $|\psi(x+x_0,t)|^2 $, and
  \begin{align}
    \label{eq:3}
    \entropyS_{x+x_0}&=-\kernBeforeIntegral\int_{-\infty}^{\infty} \diff x\,  |\psi(x+x_0,t)|^2 \ln |\psi(x+x_0,t)|^2\\
             &=-\kernBeforeIntegral\int_{-\infty}^{\infty} \diff (x+x_0)\,  |\psi(x+x_0,t)|^2 \ln |\psi(x+x_0,t)|^2\\
    &= -\kernBeforeIntegral\int_{-\infty}^{\infty} \diff x' \, |\psi(x',t)|^2 \ln |\psi(x',t)|^2=\entropyS_{x}\,.
  \end{align}
     \item[(ii)]   translations along $p$ by any $p_0$, because applying
     $T_X(p_0)$ to $\psi(x,t)$, we get
 \begin{align}
   \psi_{p_0}(x,t) &=\bra{x}\hat T_X(p_0)\ket{\psi_t} = \int_{-\infty}^{\infty} \bra{x}\hat T_X(p_0)\ket{p}\bra{p}\ket{\psi_t} \diff p \\
   &=\int_{-\infty}^{\infty} \bra{x}\ket{p+p_0}\tilde{\phi}(p,t) \diff p\\
      &=\int_{-\infty}^{\infty}  \frac{1}{\sqrt{2\piu}}\eu^{\iu \, x\, (p+p_0)}\tilde{\phi}(p,t) \diff p
    = \psi(x,t) \,  \eu^{\iu \, x\, p_0},
 \end{align}
 implying
 \begin{equation}
   \label{eq:2}
   |\psi_{p_0}(x,t)|^2 = |\psi(x,t)|^2.
 \end{equation}
 \end{enumerate}
Thus, {$\entropyS_{x}$   is invariant under  translations along $p$}.

 Similarly, by applying both translations to $\entropyS_{\mathrm{p}}=- \kernBeforeIntegral\int_{-\infty}^{\infty} \diff p \, |\tilde{\phi}(p,t)|^2 \ln |\tilde{\phi}(p,t)|^2$ we conclude that $\entropyS_{\mathrm{p}}$ is invariant under them too.

 Therefore $\entropyS=\entropyS_{x}+\entropyS_{\mathrm{p}}-3\ln \hbar$ is invariant under  translations in both $x$ and  $p$.
 \end{proof}
The same result applies in QFT setting, where $\Psi_{k_0}(x,t)=\hat T_X(k_0) \Psi(x,t)=\Psi(x,t)\,  \eu^{\iu \, x\, k_0}$ with a Fourier transform $\Phi(k+k_0,t)$.

\subsection{Entropy Invariance Under CPT Transformations}
\label{subsec:cpt}
We now show that the entropy is invariant under the  three discrete symmetries   associated with Parity Change, Time Reversal, and Charge Conjugation.

We study these symmetries in a QFT setting, such as in QED, QCD, Weak Interactions, the Standard Model, or Wightman axiomatic formulation of QFT \cite{wightman1976hilbert}. We will be focusing on fermions, and thus on the Dirac spinors equation, though most of the ideas apply to bosons as well.
A brief review of these symmetries and some properties is given in  Appendix~\ref{sec:CPT}.

\begin{definition}[C,P,T-states]
  \label{def:cpt-states}
  We denote  the quantum fields $ \Psi^\mathrm{T}(\mathbf{r},-t)\eqdefA T\Psi^*(\mathbf{r},-t)$, $\Psi^\mathrm{P}(-\mathbf{r},t)\eqdefA P\Psi(-\mathbf{r},t)$, $\Psi^{\mathrm{C}}(\mathbf{r},t)\eqdefA C\overline{\Psi}^{\tran}(\mathbf{r},t)$,  $\psi^{\mathrm{CPT}}(-\mathbf{r},-t)\eqdefA CPT\overline{\psi}^{\tran}(-\mathbf{r},-t)$.
\end{definition}
 Each of the CPT operations is briefly reviewed in  Appendix~\ref{sec:CPT}.
\begin{lemma}[{Invariance of  the entropy under CPT-transformations}]
  \label{lemma:Entropy-invariants}
  Given a quantum field $\Psi(\mathbf{r},t)$, its Fourier transform $\Phi(\mathbf{k},t)$ and its entropy $\entropyS_t$,
the entropies of   $\Psi^{*}(\mathbf{r},t)$,   $\Psi^\mathrm{P}(-\mathbf{r},t)$, $\Psi^{\mathrm{C}}(\mathbf{r},t)$, $\Psi^\mathrm{T}(\mathbf{r},-t)$, and $\Psi^{\mathrm{CPT}}(-\mathbf{r},-t)$, and their corresponding Fourier transforms,   are all equal to $\entropyS_t$.
\end{lemma}
\begin{proof}
  The probability densities of  $\Psi^{*}(\mathbf{r},t)$,   $\Psi^\mathrm{T}(\mathbf{r},-t)$, $\Psi^\mathrm{P}(-\mathbf{r},t)$, $\Psi^{\mathrm{C}}(\mathbf{r},t)$, and  $\Psi^{\mathrm{CPT}}(-\mathbf{r},-t)$   are respectively,
 \begin{align}
   \label{eq:1}
   \rho^{*}_{\mathrm{r}}(\mathbf{r},t)&=\Psi^{\tran}(\mathbf {r},t) \Psi^*(\mathbf {r},t) =\Psi^{\dagger}(\mathbf {r},t)\Psi(\mathbf {r},t)=\rho(\mathbf{r},t)\,,
   \\
   \rho^{\mathrm{T}}_{\mathrm{r}}(\mathbf{r},-t)&=\Psi^{\tran}(\mathbf {r},t) T^{\dagger} T \Psi^*(\mathbf {r},t)  = \Psi^{\tran}(\mathbf {r},t)\Psi^*(\mathbf {r},t)= \rho_{\mathrm{r}}(\mathbf{r},t)\,,
   \\
   \rho^{\mathrm{P}}_{\mathrm{r}}(-\mathbf{r},t)&=\Psi^{\dagger}(\mathbf {r},t) (\gamma^0)^{\dagger} \gamma^0 \Psi(\mathbf {r},t)  = \Psi^{\dagger}(\mathbf {r},t)\Psi(\mathbf {r},t)= \rho_{\mathrm{r}}(\mathbf{r},t)\,,
   \\
   \rho^{\mathrm{C}}_{\mathrm{r}}(\mathbf{r},t)&=\left(\overline{\Psi}^{\tran}\right)^{\dagger}(\mathbf{r},t)  \CC^{\dagger} \CC  \overline{\Psi}^{\tran}(\mathbf{r},t)   = \overline{\Psi}^{*}(\mathbf{r},t)\overline{\Psi}^{\tran}(\mathbf{r},t) = \rho_{\mathrm{r}}(\mathbf{r},t)\,,
   \\
   \rho^{\mathrm{CPT}}_{\mathrm{r}}(-\mathbf{r},-t)&=\left(\overline{\Psi}^{\tran}\right)^{\dagger}(\mathbf {r},t)  (\CC P T)^{\dagger} (\CC P T) \overline{\Psi}^{\tran}(\mathbf {r},t)
   \\
   &
   = \overline{\Psi}^{*}(\mathbf {r},t)\overline{\Psi}^{\tran}(\mathbf {r},t)
   = \rho_{\mathrm{r}}(\mathbf{r},t)\,,
 \end{align}
 where we used  $  \CC^{\dagger} \CC= T^{\dagger}T=\Identitymatrix$. Note that because  the entropy requires an integration over the whole spatial volume, the densities $ \rho^{\mathrm{P}}(-\mathbf{r},t)$ and $\rho^{\mathrm{CPT}}(-\mathbf{r},-t)$ will be evaluated  over the same volume as $\rho(\mathbf{r},t)$. As  the densities are equal, so are the associated entropies.

 The same properties above are valid for the corresponding Fourier transform fields. Note that when Parity (P) is applied to a quantum field operator, both the spatial and momentum variables  behave as odd variables,  and when time reversal is applied, both spatial and momentum operators are time reversed. So all the derivations above are similarly valid for the Fourier  quantum fields. Thus, both entropies terms in  $\entropyS_t=\entropyS_r+\entropyS_k$ are invariant under all CPT transformations.
\end{proof}

Note that a proof that the entropy is invariant under  CPT and each of the  transformations can be readily achieved in quantum \mechanics, without the QFT framework.

\subsection{Entropy Invariance Under the Special Relativity's Lorentz Group}
\label{subsec:Lorentz-transformation}
We now investigate the behavior of the entropy when space and time are transformed according to the Lorentz transformations.  The natural setting for the Lorentz transformations is QFT, as space and time are treated similarly.
\begin{proposition}
The entropy  is a scalar in special relativity.
\end{proposition}
{
\begin{proof}
The probability elements
\begin{align}
    \diff \pr(\mathbf{r},t) =\rho_{\mathrm{r}}(\mathbf{r},t)\diff^3\mathbf{r}\qquad  \text{and} \qquad \diff \pr(\mathbf{k},t) =\rho_{\mathrm{k}}(\mathbf{k},t)\diff^3\mathbf{k}
    \label{eq:probability-elements}
\end{align}
are invariant under special relativity since probabilities of events  do not depend on the frame of reference. Consider a slice of the phase space with a given frequency  $\omega_{\mathrm{k}}= \sqrt{ \mathbf{k}^2 c^2+ \left(\frac{m c^2}{\hbar}\right )^2}$. The volume elements $\frac{1}{\omega_{\mathrm{k}}} \diff^3\mathbf{k}$ and $\omega_{\mathrm{k}} \diff^3\mathbf{r}$,  are invariant under the Lorentz group~\cite{weinberg1995quantum1}, that is,
\begin{align}
\frac{1}{\omega_{\mathrm{k}}} \diff^3\mathbf{k}=\frac{1}{\omega_{\mathrm{k}'}} \diff^3\mathbf{k}'\qquad  & \qquad \omega_{\mathrm{k}} \diff^3\mathbf{r}=\omega_{\mathrm{k}'} \diff^3\mathbf{r}'
\\
 & \Downarrow
 \\
    \diff V \eqdefA  \diff^3\mathbf{k} \diff^3\mathbf{r} & = \diff^3\mathbf{k}' \diff^3\mathbf{r}' \eqdefA  \diff V '\,,
\end{align}
where $\mathbf{r}'  $, $\mathbf{k}', $ and $ \omega_{\mathrm{k}'}$ are the results of a Lorentz transformation applied to  $\mathbf{r}$, $\mathbf{k}$, and $\omega_{\mathrm{k}}$.
 Thus, $\frac{1}{\omega_{\mathrm{k}}}\rho_{\mathrm{r}}(\mathbf{r},t)$ and $\omega_{\mathrm{k}} \rho_{\mathrm{k}}(\mathbf{k},t)$ are invariant under the Lorentz group.
Thus, the phase space density $\rho_{\mathrm{r}}(\mathbf{r},t) \rho_{\mathrm{k}}(\mathbf{k},t)$ is an invariant under the Lorentz group. Therefore the entropy  is invariant in special relativity, that is, it is a scalar quantity under Lorentz transformations.
\end{proof}
Note that in QFT, one does scale the  operator  $\Phi(\mathbf{k},t)$ by $\sqrt{2\omega_{\mathrm{k}}}$, that is, one scales the creation and the annihilation operators ${\alpha}^{\dagger}(\mathbf{k})=\sqrt{\omega_{\mathrm{k}}}\, \mathbf{a}^{\dagger}(\mathbf{k})$ and ${\alpha}(\mathbf{k})=\sqrt{\omega}\,  \mathbf{a}(\mathbf{k})$. In this way,  the density operator $\Phi^{\dagger}(\mathbf{k},t)\Phi(\mathbf{k},t)$
scales with $\omega_{\mathrm{k}}$ and becomes  an invariant in  special relativity. Also, with such scaling, the  infinitesimal  probability of finding a particle with momentum $\mathbf{p}=\hbar \mathbf{k}$ in the original reference frame is invariant  under Lorentz transformation, though the same particle would be found with  momentum $\mathbf{p}'=\hbar \mathbf{k}'$.
}

\subsection{\QCurves and Entropy-Partition}
\label{subsec:entropy-partition}

We define a \QCurve to be a curve (or path) in Hilbert space representing the time evolution of all degrees of freedom of an initial state,  according to a Hamiltonian,  for a time interval.  In quantum \mechanics we can represent it by  a triple $\big(\ket{\psi_{0}}, U(t), [t_{0},t_{1}] \big)$  where $\ket{\psi_{0}}$ is the initial state at time $t_{0}$, $U(t)=\eu^{-\iu H t}$ is the evolution operator, and $[t_{0},t_{1}]$ the time interval of the evolution. Alternatively, we can represent the initial state of  a \QCurve in a phase space, via $\ket{\psi_{0}}\mapsto \bra{\mathbf{r}}\ket{\psi_{0}}, \bra{\mathbf{p}}\ket{\psi_{0}}$.

\QCurves representation can readily be expanded to a QFT setting where one defines the initial state as $\left ( \Psi(\mathbf{r},0)\ket{0}, \Phi(\mathbf{k},0)\ket{0}\right)$ or simply $\left (\Psi_0(\mathbf{r}), \Phi_0(\mathbf{k})\right) $ and the \QCurve evolution is described in the QFT phase space.
We will adapt here the quantum \mechanics representation but all the definitions in this section can be straightforwardly  written in the QFT setting.

For manipulations purposes we may also write the  \QCurve triple as $\big(\ket{\psi_{0}}, U(t), [0, t_{1}-t_{0}] \big)$ to have the starting value $t=0$, or if that simplifies,   $\big(\ket{\psi_{0}}, U(t), \delta t \big)$, where $\delta t= t_{1}-t_{0}$.
To allow unbounded evolutions, we could set $t_{1}=\infty$, and then $[t_{0},\infty]$ would stand for $[t_{0},\infty)$.

\begin{definition}[Partition of $\mathscr{E}$]
  \label{def:partition}
Let $\setAll$ to be the set of all \QCurves.  We define a partition of  $\setAll$ based on the entropy evolution into four blocks   as follows
\begin{enumerate}
 \item
$\setConstant$ is the set of the \QCurves for which the entropy is a constant.
 \item
 $\setDecreasing$ is the set of the \QCurves for which the entropy is decreasing, but it is not a constant.
 \item
 $\setIncreasing$ is the set of  the \QCurves for which the entropy is increasing, but it is not a constant.
 \item
$\setOscillating$ is the set of the remaining \QCurves.
 \end{enumerate}
 \end{definition}

 The \QCurves in $\setOscillating$ are oscillating, with the entropy strictly decreasing in some subinterval of $[t_{0},t_{1}]$ and strictly increasing in another subinterval of $[t_{0},t_{1}]$. Clearly, the \QCurves in $\setOscillating$ are concatenations of \QCurves in $\setDecreasing \bigcup \setIncreasing$ of shorter durations.

\begin{definition}[\QCurve Reflection]
  Given a  \QCurve $ \big(\ket{\psi_{0}}, U(t), [t_{0}-t_{1}] \big)$, a \QCurve $\big(\ket{\psi'_{0}}, U(t), [t'_{0},t'_{1}] \big)$ is a reflection  of it if and only if $t_1 - t_0 = t'_1 - t'_0 = \deltau t$ and for every $\tau \in [0, \deltau t]$, $\entropyS\big( \ket{\psi_{t_0 + \tau}} \big ) = \entropyS\left( \ket{\psi'_{t'_1-\tau}} \right )$.
\end{definition}
It is straightforward to verify  that if a  a \QCurve $e'_0$ is a reflection of a \QCurve $e_0$, then $e_0$ is a reflection of $e'_0$, that is the reflection of \QCurves is an involution.

Stated informally, if the  two \QCurves are shifted so that their evolution intervals are aligned, they are the reflections of each other at the midpoint of the evolution interval. Note also that if the entropy of the first curve is increasing, the entropy of the second curve is decreasing, and vice versa.

\begin{definition}[Set Reflection]
  \label{def:ReflectionSets}
We define a symmetric reflection relation on the subsets of $\setAll$.
 $\setAll_{2}$ is a reflection of $\setAll_{1}$ if and only if for every \QCurve $e_1\in\setAll_{1}$ there is a \QCurve  $e_2\in\setAll_{2}$ that is  a  reflection of $e_1$, and vice versa.
\end{definition}

\section{Dispersion Transform: The Entropy of Coherent States Increases with Time}
\label{sec:dispersion-coherent-states}

This section is focused on the spatial distribution of a state and not on its internal degrees of freedom. The section can be read in a  quantum \mechanics setting  or a QFT setting since both space and momentum are simply treated as Fourier  transforms of each other. The vector notation used here is of a quantum \mechanics setting, but to suggest the connections we use $\mathbf {k}=\frac{1}{\hbar}\mathbf {p}$ representation to emphasize the similarity of the two representations for this section, and so $\entropyS_{\mathrm{k}}=\entropyS_{\mathrm{p}}-3\ln \hbar$.

We now consider initial solutions that are localized in space, $\psi_{\mathrm{k}_0}(\mathbf {r}-\mathbf {r}_0)=\psi_0(\mathbf {r}-\mathbf {r}_0) \, \eu^{\iu \mathbf {k}_0\cdot \mathbf {r}}$, where $\mathbf{r}_0$ is the mean value of $\mathbf {r}$ according to the probability distribution $\rho_{\mathrm{r}}(\mathbf {r})=|\psi_0(\mathbf {r})|^2$, and the phase term $\eu^{\iu \mathbf {k}_0\cdot \mathbf {r}}$ gives the momentum shift of $\mathbf {k}_0$. Assume that the variance, $\sigma^2 =\int \diff^3\mathbf {r}\, (\mathbf {r}-\mathbf {r}_0)^2 \rho(\mathbf {r}) $, is finite.  $\psi_{\mathrm{k}_0}(\mathbf {r}-\mathbf {r}_0)$ evolves according to the given Hamiltonian with a dispersion relation $\omega(\mathbf {k})$. In a Cartesian representation, we can write the initial state in the momentum space as $ \phi_{\mathrm{r}_0}(\mathbf {k}-\mathbf {k}_0)= \phi_{0}(\mathbf {k}-\mathbf {k}_0) \, \eu^{-\iu (\mathbf {k}-\mathbf {k}_0)\cdot \mathbf {r}_0}$, where $\phi_{0}(\mathbf {k})$ is the Fourier transform of $\psi_0(\mathbf {r})$, and so
 $\rho_{\mathrm{k}}(\mathbf {k})=|\phi_{\mathrm{r}_0}(\mathbf {k}-\mathbf {k}_0)|^2$ also has a finite variance, $\sigma_{\mathrm{k}}^2$,  with the mean in the  momentum center
$\mathbf {k}_0$. Evolving the wave function in  momentum space according to the dispersion relation $\omega(\mathbf {k})$ and taking the inverse Fourier transform, we get
\begin{align}
 \psi(\mathbf {r}-\mathbf {r}_0,t) &={\frac {1}{({\sqrt {2\piu }})^{3}}}\int \Phi_{\mathrm{r}_0} (\mathbf {k}-\mathbf {k}_0) \, \eu^{-\iu \omega(\mathbf {k}) t} \eu^{\iu \mathbf {k} \cdot \mathbf {r} }\diff^{3}\mathbf {k}\,.
 \label{eq:time-evolution-psi-general}
\end{align}
As $\phi_{\mathrm{r}_0}(\mathbf {k}-\mathbf {k}_0)$ fades away exponentially from $\mathbf {k}=\mathbf {k}_0$, we can expand the dispersion formula in a Taylor series and approximate it as
\begin{align}
 \omega(\mathbf {k})   &\approx
 \mathbf{\mathbf{ v_p}}(\mathbf {k}_0) \cdot \mathbf {k}_0+\mathbf{v_g}(\mathbf {k}_0) \cdot (\mathbf {k}-\mathbf{k}_0)
+
 \frac{1}{2}(\mathbf {k}-\mathbf {k}_0)^{\tran} \hessian(\mathbf {k}_0)\,
 (\mathbf {k}-\mathbf {k}_0)\, , \,
  \label{eq:w-dispersion}
\end{align}
where $\mathbf{\mathbf{ v_p}}(\mathbf {k}_0)$, $\mathbf{ \mathbf{v_g}}(\mathbf {k}_0)$, and $ \hessian(\mathbf {k}_0) $ are the phase velocity, the group velocity, and the Hessian of the dispersion relation $\omega(\mathbf {k})$, respectively. Then after inserting~\eqref{eq:w-dispersion}    into~\eqref{eq:time-evolution-psi-general},
\begin{flalign}
 \psi_{\mathrm{k}_0}(\mathbf {r}-\mathbf {r}_0-\mathbf{v_g}(\mathbf {k}_0) t,t)  & \approx
 \frac{1}{Z_{\mathrm{r}}}\eu^{-\iu t \mathbf{\mathbf{ v_p}}(\mathbf {k}_0) \cdot \mathbf {k}_0} \,  \psi_{\mathrm{k}_0}(\mathbf {r}- (\mathbf {r}_0 +\mathbf{v_g}(\mathbf {k}_0) t))\\
 &\quad \ast \normalx{\mathbf {r}}{\mathbf{r}_0+\mathbf{v_g}(\mathbf {k}_0) t}{\iu t \hessian(\mathbf {k}_0) }\,,
 \\
 \Phi_{\mathrm{r}_0+\mathbf{v_g}(\mathbf {k}_0)\, t} (\mathbf {k}-\mathbf {k}_0, t) & \approx   \frac{1}{Z_k} \, \eu^{-\iu t \mathbf{\mathbf{ v_p}}(\mathbf {k}_0) \cdot \mathbf {k}_0} \, \Phi_{\mathrm{r}_0+\mathbf{v_g}(\mathbf {k}_0)\, t} (\mathbf {k}-\mathbf {k}_0)  \\
 &\quad \times \normalx{\mathbf {k}}{\mathbf{k}_0}{-\iu t^{-1} \hessian^{-1}(\mathbf {k}_0) }\, ,
 \label{eq:time-evolution-psi-dispersion}
\end{flalign}
where $\psi_{\mathrm{k}_0}(\mathbf {r}-\mathbf {r}_0-\mathbf{v_g}(\mathbf {k}_0) t,t)$ is the spatial Fourier transform of $\Phi_{\mathrm{r}_0+\mathbf{v_g}(\mathbf {k}_0)\, t} (\mathbf {k}-\mathbf {k}_0, t) $, $\ast$ denotes a convolution, $Z_{r,k}$ normalizes the amplitude, and $\mathscr{N}$ denotes a normal distribution. We can interpret this evolution as describing a wave moving with phase velocity $\mathbf{\mathbf{ v_p}}(\mathbf {k}_0)$, group velocity $\mathbf{v_g}(\mathbf {k}_0)$, and being blurred by a time varying  imaginary-valued symmetric matrix $\iu t \hessian(\mathbf {k}_0) $.

We refer to \eqref{eq:time-evolution-psi-dispersion}  as the \textit{quantum dispersion transform}.

The probability density functions associated with the probability amplitude functions in \eqref{eq:time-evolution-psi-dispersion} are
\begin{align}
 \rho_{\mathrm{r}}(\mathbf {r}-(\mathbf {r}_0+\mathbf{v_g}(\mathbf {k}_0) \, t),t) & = \frac{1}{Z_{\mathrm{r}}^2}
 | \psi_{\mathrm{k}_0}(\mathbf {r}- (\mathbf {r}_0 +\mathbf{v_g}(\mathbf {k}_0) \, t)) \\
 & \quad \quad \ast \normalx{\mathbf {r}}{\mathbf{r}_0+\mathbf{v_g}(\mathbf {k}_0) \, t}{\iu \, t\, \hessian(\mathbf {k}_0) }|^2\,,
 \\
\rho_{\mathrm{k}}(\mathbf {k}-\mathbf {k}_0,t) & = \frac{1}{Z_k^2}|\Phi_{\mathrm{r}_0+\mathbf{v_g}(\mathbf {k}_0)\, t} (\mathbf {k}-\mathbf {k}_0)|^2   \, .
 \label{eq:time-evolution-rho-dispersion}
\end{align}

\begin{lemma}[Dispersion Transform and Reference Frames]
\label{lemma:simplified-densities}
The entropy associated with the probability amplitudes \eqref{eq:time-evolution-psi-dispersion} and their respective probability densities \eqref{eq:time-evolution-rho-dispersion}  is the same as the entropy associated with the simplified probability densities
\begin{align}
 \rho^\entropyS_{\mathrm{r}}(\mathbf {r},t)&=  \frac{1}{Z^2}| \psi_{0}(\mathbf {r}) \ast \normalx{\mathbf {r}}{0}{\iu \, t\, \hessian(0) }|^2\,,
 \\
\rho^\entropyS_{\mathrm{k}}(\mathbf {k},t) & = \frac{1}{Z_k^2}|\Phi_{0} (\mathbf {k})|^2 =\rho^\entropyS_{\mathrm{k}}(\mathbf {k},t=0)\, .
 \label{eq:rho-dispersion-transform-simplified}
\end{align}
\end{lemma}
\begin{proof}

Consider the probability densities \eqref{eq:time-evolution-rho-dispersion}. If the frame of reference is translating the position by $\mathbf {r}_0'=\mathbf {r}_0+\mathbf{v_g}(\mathbf {k}_0) \, t$ and the momentum by $k_0$, that will yield the  simplified density functions \eqref{eq:rho-dispersion-transform-simplified}.

By Lemma~\ref{lemma:frame-invariance}, the  entropy in position and momentum is  invariant under translations of the position $\mathbf{r}$ and  the momentum $\mathbf{k}$.
\end{proof}

The time invariance of the density $\rho^\entropyS_{\mathrm{k}}(\mathbf {k},t)$, and therefore of $\entropyS_{\mathrm{p}}$, reflects  the conservation law of momentum for free particles.

We now consider coherent states, $\ket{\alpha}$, the eigenstates of the annihilator operator used to define quantum fields. In position and momentum space they are
\begin{align}
 \psi_{\mathrm{k}_0}(\mathbf {r}-\mathbf {r}_0)& = \bra{\mathbf {r}}\ket{\alpha}=\frac{1}{2^3 \piu^{\frac{3}{2}}(\det \matrixsym{\Sigmabold})^{\frac{1}{2}}}\, \normalx{\mathbf {r}}{\mathbf {r}_0}{\matrixsym{\Sigmabold} }\, \eu^{\iu \mathbf {k}_0\cdot \mathbf {r}}\,,
 \\
 \Phi_{\mathrm{r}_0} (\mathbf {k}-\mathbf {k}_0) &=\bra{\mathbf {k}}\ket{\alpha}=\frac{1}{2^3 \piu^{\frac{3}{2}}(\det \matrixsym{\Sigmabold}^{-1})^{\frac{1}{2}}}\, \normalx{\mathbf {k}}{\mathbf {k}_0}{\matrixsym{\Sigmabold}^{-1} }\, \eu^{\iu (\mathbf {k}-\mathbf {k}_0)\cdot \mathbf {r}_0}\,,
 \label{eq:coherent-state-3D}
\end{align}
where   $ \matrixsym{\Sigmabold} $ is the spatial covariance matrix. They are neither eigenstates of the  free-particle \Schroedinger equation nor of the Dirac equation.

\begin{theorem}
\label{thm:coherent-entropy}
A \QCurve with  an initial coherent state and  evolving according to  a free-particle \Schroedinger or Dirac
Hamiltonian is in \setIncreasing.
\end{theorem}
\begin{proof}
To describe the evolution of the initial states \eqref{eq:coherent-state-3D}, we apply the quantum dispersion transform \eqref{eq:time-evolution-psi-dispersion}.
According to Lemma~\ref{lemma:simplified-densities},
the entropy associated with the resulting probability densities is the same as the one associated with the simplified probability densities
\begin{align}
\rho^{\text{S}}_{\mathrm{r}}(\mathbf {r},t)& = \frac{1}{Z^2}| \normalx{\mathbf {r}}{0}{\matrixsym{\Sigmabold} } \ast \normalx{\mathbf {r}}{0}{\iu \, t\, \hessian(0) }|^2 \\
 &=\frac{1}{Z_2^2} \normalx{\mathbf {r}}{0}{\matrixsym{\Sigmabold} +\iu t \hessian} \normalx{\mathbf {r}}{0}{\matrixsym{\Sigmabold} -\iu t \hessian}
                                               = \normalx{\mathbf {r}}{0}{\frac{1}{2}\matrixsym{\Sigmabold}(t)}\,,
 \\
 \rho^{\text{S}}_{\mathrm{k}}(\mathbf {k},t) & = \normalx{\mathbf {k}}{0}{\matrixsym{\Sigmabold}^{-1} }\,,
\end{align}
where $\matrixsym{\Sigmabold}(t)=\matrixsym{\Sigmabold} + t^2 \hessian \matrixsym{\Sigmabold}^{-1} \hessian$.
Then
\begin{align}
 \entropyS & =\entropyS_{\mathrm{r}}+\entropyS_{\mathrm{k}}
 \\
            &= {-\kernBeforeIntegral\int_{\Omega} \normalx{\mathbf {r}}{0}{\frac{1}{2}\matrixsym{\Sigmabold}(t)} \ln \normalx{\mathbf {r}}{0}{\frac{1}{2}\matrixsym{\Sigmabold}(t)}\diff^3\mathbf {r}} \\
  & \qquad {}-\kernBeforeIntegral\int_{\Omega} \normalx{\mathbf {k}}{0}{\matrixsym{\Sigmabold}^{-1}} \ln \normalx{\mathbf {k}}{0}{2\matrixsym{\Sigmabold}^{-1}}\diff^3\mathbf {k}
 \\
 &= 3 (1+  \ln \piu )  + \frac{1}{2}\ln \det\left ( \Identitymatrix+ t^2 (\matrixsym{\Sigmabold}^{-1} \hessian)^2\right)\,  .
 \label{eq:entropy-coherent}
\end{align}
It is easy to see that $\det\left ( \Identitymatrix+ t^2 (\matrixsym{\Sigmabold}^{-1} \hessian)^2\right)>0$,
and therefore the entropy increases over time.
\end{proof}
Note that at $t=0$ coherent states have minimum entropy.
The theorem suggests that quantum theory has an inherent mechanism to increase entropy for free particles, the root cause being the position dispersion property of the Hamiltonian.

One possible thought at this point is the following.
Coherent states are composed of superposition of stationary states, and
as we show in Proposition~\ref{prop:stationaryStatesAreInC}, stationary solutions have constant entropy.
One may then attempt to draw a parallel between mixing stationary states of a Hamiltonian in quantum \mechanics  and mixing different ideal gases in equilibrium in  statistical \mechanics.  It may appear that in both scenarios, if the states are not mixed the entropy  is constant over time; but, after mixing, the entropy increases. However, as we will show,
in quantum physics the entropy may  decrease during state evolution.

\section{Entropy Oscillation}
\label{sec:oscillation}

This section addresses several scenarios when \QCurves are in \setOscillating.  We revisit Fermi's golden rule foundations and present an exact derivation of the coefficients used for the rule. We then  explore time reversals accompanied by time translations.  We conclude by studying the oscillation induced by a collision of two particles, each one described by a coherent state, and we show that the oscillation arises due to the interference of the two particles, that is their entanglement.

\subsection{Fermi's Golden Rule Revisited}
\label{subsec:Fermi}
The notation of this section is in the quantum \mechanics setting, but it can be easily adapted to the QFT setting.
Let us consider stationary states $\ket{\psi_t}=\ket{\psi_E} \ee^{-\iu \omega t} $ with $\omega={E}/{\planckbar}$, where  $E$ is the energy eigenvalue of the Hamiltonian, and $\ket{\psi_E}$ is the time-independent eigenstate of the Hamiltonian  associated with energy $E$.
\begin{proposition}
\label{prop:stationaryStatesAreInC}
All  stationary states  are in  $\setConstant$.
\end{proposition}
\begin{proof}
We evaluate
\begin{align}
  \rho_{\mathrm{r}}(\mathbf{r},t)&=|\bra{\mathbf{r}}\ket{\psi_t} |^2 = |\bra{\mathbf{r}}\ket{\psi_{E_i}} \ee^{-\iu \omega_i t}|^2=  |\bra{\mathbf{r}}\ket{\psi_{E_i}}|^2=\rho_{\mathrm{r}}(\mathbf{r},0) \,,
  \\
   \rho_{\mathrm{k}}(\mathbf{k},t)&=|\bra{\mathbf{p}}\ket{\psi_t} |^2 =  |\bra{\mathbf{p}}\ket{\psi_{E_i}} \ee^{-\iu \omega_i t}|^2= |\bra{\mathbf{p}}\ket{\psi_{E_i}}|^2=\rho_{\mathrm{k}}(\mathbf{k},0) \,. \end{align}
 Both are time invariant, and thus the entropy  is time invariant.
\end{proof}

\begin{theorem}[Fermi's golden rule coefficients for two states]
\label{theorem:FermiCoefficientsExactTwoState}
Consider a  particle in  an eigenstate  $\ket{\psi_{E_1}}$ of  a Hamiltonian $H^0$ that has only two eigenstates $\ket{\psi_{E_1}}$ and $ \ket{\psi_{E_2}}$  with eigenvalues $E_1=\hbar \omega_1$ and $E_2=\hbar \omega_2$, respectively. Let this particle interact with an external field (such as the impact of a  Gauge Field), requiring an additional Hamiltonian term $H^{\mathrm{I}}$ to describe the evolution of this system.

The evolution of this particle is described by $\ket{\psi_{t}}=\ee^{-\iu \frac{(H^0+H^{\mathrm{I}})}{\hbar} t} \ket{\psi_{E_1}}=\alpha_1(t)\ket{\psi_{E_1}}+\alpha_2(t)\ket{\psi_{E_2}}$, with $\alpha_1(0)=1$ and $\alpha_2(0)=0$.

Let $\omega_{11}\eqdefA  \omega_1 + \omega_{11}^{\mathrm{I}}$, $\omega_{22}\eqdefA  \omega_1 + \omega_{22}^{\mathrm{I}}$, and   $\omega^{\mathrm{I}}_{i,j}\eqdefA \frac{1}{\hbar}\bra{\psi_{E_i}}H^{\mathrm{I}}\ket{\psi_{E_j}}$. The probability of this particle to be in  state $\ket{\psi_{E_2}}$ at time $t$ is given by

$|\alpha_2(t)|^2 = \frac{4(\omega_{12}^{\mathrm{I}})^2  }{(\omega_1+\omega_{11}^{\mathrm{I}}-\omega_2-\omega_{22}^{\mathrm{I}})^2+4(\omega_{12}^{\mathrm{I}})^2} \, \sin^2 \frac{(\lambda_2-\lambda_1) t}{2} $,

where $\lambda_{1,2} = \frac{\omega_1+\omega_{11}^{\mathrm{I}}+\omega_2+\omega_{22}^{\mathrm{I}} \pm \sqrt{(\omega_1+\omega_{11}^{\mathrm{I}}-\omega_2-\omega_{22}^{\mathrm{I}})^2+4 (\omega_{12}^{\mathrm{I}})^2}}{2}$.
\end{theorem}
\begin{proof}
The Hamiltonian can be written in the basis $\ket{\psi_{E_1}}, \ket{\psi_{E_2}}$ as
\begin{align}
    H^0 = \hbar \begin{pmatrix} \omega_1 & 0 \\ 0 & \omega_2 \end{pmatrix}\qquad \text{and} \qquad  H^{\mathrm{I}} = \hbar \begin{pmatrix} \omega_{11}^{\mathrm{I}} & \omega_{12}^{\mathrm{I}} \\ \omega_{12}^{\mathrm{I}} & \omega_{22}^{\mathrm{I}} \end{pmatrix}\,,
\end{align}
where the real values satisfy  $\omega_{21}^{\mathrm{I}}=\omega_{12}^{\mathrm{I}}$ to assure that $H^{\mathrm{I}}$ is Hermitian.
The time evolution of the state $ \ket{\psi_{E_1}}$ is described by
\begin{align}
     \ket{\psi_t}&=\ee^{-\iu \frac{(H^0+H^{\mathrm{I}})}{\hbar} t} \ket{\psi_{E_1}}
      =  \sum_{k=1}^2 \alpha_k(t) \ket{\psi_{E_k}}
    \\
    & \Downarrow \qquad \text{projecting on \, } \bra{\psi_{E_j}}
    \\
    \alpha_j(t)   &=  \bra{\psi_{E_j}}\ee^{-\iu \frac{(H^0+H^{\mathrm{I}})}{\hbar} t} \ket{\psi_{E_1}}\,.
\end{align}
As $H^0+H^{\mathrm{I}}$ is symmetric, the matrix of the orthonormal eigenvectors $\bf{v_1}$ , $\bf{v_2}$ can be described by
\begin{align}
    \begin{pmatrix} \bf{v_1} & \bf{v_2}\end{pmatrix}=R(\theta) =\begin{pmatrix}
        \cos \theta &  -\sin \theta
        \\
         \sin \theta &  \cos \theta
    \end{pmatrix}\,,
\end{align}
where $R(\theta)$ is a rotation matrix,  and
\begin{align}
\label{eq:theta-definition}
    \sin 2\theta &=\frac{2\omega_{12}^{\mathrm{I}}}{\sqrt{(\omega_{11}-\omega_{22})^2+4(\omega_{12}^{\mathrm{I}})^2}}\,,
    \\
    \cos 2\theta &=\frac{\omega_{11}-\omega_{22}}{\sqrt{(\omega_{11}-\omega_{22})^2+4(\omega_{12}^{\mathrm{I}})^2}}\,,
\end{align}
with eigenvalues
\begin{align}
\label{eq:eigenvalues-two-state}
    \lambda_{1,2} &= \frac{\omega_{11}+\omega_{22} \pm \sqrt{(\omega_{11}-\omega_{22})^2+4 (\omega_{12}^{\mathrm{I}})^2}}{2}\,.
\end{align}
Therefore,
\begin{align}
   H^0+H^{\mathrm{I}} &= \hbar \begin{pmatrix} \omega_{11} & \omega_{12}^{\mathrm{I}} \\ \omega_{12}^{\mathrm{I}} & \omega_{22} \end{pmatrix}
   =\begin{pmatrix}
        \cos \theta &  -\sin \theta
        \\
         \sin \theta &  \cos \theta
    \end{pmatrix} \begin{pmatrix}
      \hbar \lambda_{1} &  0
        \\
         0 &  \hbar\lambda_{2}
    \end{pmatrix} \begin{pmatrix}
        \cos \theta &  \sin \theta
        \\
         -\sin \theta &  \cos \theta
    \end{pmatrix}\, .
    \\
    & \Downarrow
    \\
   \ee^{-\iu \frac{(H^0+H^{\mathrm{I}})}{\hbar} t} &=\begin{pmatrix}
        \cos \theta &  -\sin \theta
        \\
         \sin \theta &  \cos \theta
    \end{pmatrix} \begin{pmatrix}
     \ee^{-\iu  \lambda_1 t} &  0
        \\
         0 &  \ee^{-\iu  \lambda_2 t}
    \end{pmatrix} \begin{pmatrix}
        \cos \theta &  \sin \theta
        \\
         -\sin \theta &  \cos \theta
    \end{pmatrix}
    \\
    &=\begin{pmatrix}
         \ee^{-\iu  \lambda_1 t}\cos^2 \theta  + \ee^{-\iu  \lambda_2 t}\sin^2 \theta & \,  \frac{\ee^{-\iu  \lambda_1 t} - \ee^{-\iu  \lambda_2 t}}{2} \sin 2 \theta
        \\
         \frac{\ee^{-\iu  \lambda_1 t} - \ee^{-\iu  \lambda_2 t}}{2} \sin 2 \theta & \,   \ee^{-\iu  \lambda_1 t}\sin^2 \theta  + \ee^{-\iu  \lambda_2 t}\cos^2 \theta
    \end{pmatrix}\,.
\end{align}

Thus, the coefficients are
\begin{align}
\label{eq:Fermi-coefficient-two-state}
     \begin{pmatrix}
    \alpha_1(t)
    \\
    \alpha_2(t)
    \end{pmatrix} &=  \ee^{-\iu \frac{(H^0+H^{\mathrm{I}})}{\hbar} t}   \begin{pmatrix}
    1
    \\
    0
    \end{pmatrix}=\begin{pmatrix}
         \cos^2 \theta \, \ee^{-\iu  \lambda_1 t}  + \sin^2 \theta \,  \ee^{-\iu  \lambda_2 t}
        \\
         \sin 2 \theta \,  \left(\frac{\ee^{-\iu  \lambda_1 t} - \ee^{-\iu  \lambda_2 t}}{2} \right)
    \end{pmatrix}
    \\
    & \Downarrow
    \\
    \begin{pmatrix}
    |\alpha_1(t)|^2
    \\
    |\alpha_2(t)|^2
    \end{pmatrix} &= \begin{pmatrix}
     1-\frac{1}{2} \sin^2 2\theta \left (1-  \cos (\lambda_2-\lambda_1) t\right)
    \\
  \frac{1}{2}  \sin^2 2 \theta \, \left (1-  \cos (\lambda_2-\lambda_1) t\right)
    \end{pmatrix}\,.
\end{align}
Noting that $1-  \cos (\lambda_2-\lambda_1) t=2\, \sin^2 \frac{(\lambda_2-\lambda_1) t}{2}$ we see that the probability of being in  state $\ket{\psi_{E_2}}$ at time $t$  is $|\alpha_2(t)|^2=\sin^2 2 \theta \, \sin^2 \frac{(\lambda_2-\lambda_1) t}{2}$.
\end{proof}

If $\omega_{1} \gg \omega_{11}$, $\omega_{2} \gg \omega_{22}$, and $|\omega_{1}- \omega_{2}| \gg \omega_{12}$, then $\lambda_{1,2}= \frac{\omega_{11}+\omega_{22} \pm \sqrt{(\omega_{11}-\omega_{22})^2+4 (\omega_{12}^{\mathrm{I}})^2}}{2}\approx \omega_{1,2}$, and the coefficient of transition becomes $|\alpha_2(t)|^2\approx \frac{4(\omega_{12}^{\mathrm{I}})^2  }{(\omega_{1}-\omega_{2})^2} \, \sin^2 \frac{(\omega_2-\omega_1) t}{2}$, which is the usual approximation made when using Fermi's golden rule.

Our derivation is  quite different from the ones found in the literature, and in particular we start from an exact calculation of the coefficient and when approximating there is no need to restrict to the case where the time interval is small.

Also note that it is straightforward to derive the final state if the initial state is in  a  superposition  $\ket{\psi_0}=\alpha_1^0\ket{\psi_{E_1}}+\alpha_2^0 \ket{\psi_{E_2}}$, with $1=|\alpha_1^0|^2+|\alpha_2^0|^2$. In this case,
\begin{align}
     \ket{\psi_t}&=\ee^{-\iu \frac{(H^0+H^{\mathrm{I}})}{\hbar} t} \ket{\psi_0}
      =  \sum_{k=1}^2 \alpha_k(t) \ket{\psi_{E_k}}
    \\
    & \Downarrow
    \\
   \begin{pmatrix}
    \alpha_1(t)
    \\
    \alpha_2(t)
    \end{pmatrix} &=\ee^{-\iu \frac{(H^0+H^{\mathrm{I}})}{\hbar} t}  \begin{pmatrix}
    \alpha_1^0
    \\
    \alpha_2^0
    \end{pmatrix}= \begin{pmatrix}
         \alpha_1^0 (\cos^2 \theta \, \ee^{-\iu  \lambda_1 t}  + \sin^2 \theta \,  \ee^{-\iu  \lambda_2 t})+\alpha_2^0 \sin 2 \theta \,  \left(\frac{\ee^{-\iu  \lambda_1 t} - \ee^{-\iu  \lambda_2 t}}{2} \right)
        \\
         \alpha_1^0 \sin 2 \theta \,  \left(\frac{\ee^{-\iu  \lambda_1 t} - \ee^{-\iu  \lambda_2 t}}{2} \right)+\alpha_2^0 \left (\ee^{-\iu  \lambda_1 t}\sin^2 \theta  + \ee^{-\iu  \lambda_2 t}\cos^2 \theta \right)
    \end{pmatrix}\,,
\end{align}
where Theorem~\ref{theorem:FermiCoefficientsExactTwoState} is applied to the case of $\alpha_1^0=1$ and $\alpha_2^0=0$.

For the general case of $N$ eigenstates, we can compute the transition from $\ket{\psi_{E_1}}$ to any state $\ket{\psi_{E_j}}$ at time $t$. First one must retrieve the eigenvalues $\lambda_i, i=1,\hdots, N$ and normalized eigenstates $v_i, i=1,\hdots, N$ of the symmetric $N\times N$  matrix $ H^0+H^{\mathrm{I}}$. One method requires finding one eigenvalue and eigenstate and then uses a Gram-Schmidt process to derive the other eigenstates. Then the following simple procedure recovers all the coefficients and their probability of transitions from $\ket{\psi_{E_1}}$ to any state $\ket{\psi_{E_j}}$
\begin{align}
\alpha_j & = \bra{\psi_{E_j}}\ee^{-\iu \frac{(H^0+H^{\mathrm{I}})}{\hbar} t}\ket{\psi_{E_1}}
=\sum_{i=1}^N \ee^{-\iu \lambda_i t} v_{ij} v_{i1}
\\
&\Downarrow
\\
|\alpha_j|^2 &=\sum_{i=1}^N \sum_{k=1}^N  \ee^{-\iu \lambda_i t} v_{ij} v_{i1}\ee^{\iu \lambda_k t} v_{kj} v_{k1}
\\
&=\sum_{i=1}^N \left [  v^2_{ij} v^2_{i1} + 2 \sum_{k>i}^N  v_{ij} v_{i1} v_{kj} v_{k1} \cos (\lambda_i-\lambda_k) t \right]\,.
\end{align}
This is an exact evaluation of the transition coefficients. Approximations to obtain Fermi's golden rule can then be made similarly to the approximations we made above for the case of two states.

\begin{theorem}[Entropy Oscillations]
\label{theorem:oscillations}
Consider a  particle in  an eigenstate $\ket{\psi_{E_1}}$ of  a Hamiltonian $H^0$ that has only two eigenstates, $\ket{\psi_{E_1}} $ and $  \ket{\psi_{E_2}}$, with eigenvalues $E_1$ and $E_2$ respectively. Let this particle interact with an external field (such as a Gauge Field), requiring an additional Hamiltonian term $H^{\mathrm{I}}$ to describe the evolution of this system.

Then, for a time interval such that  $\deltau t > \frac{\piu }{|\lambda_1-\lambda_2|}$, where $\lambda_{1,2}$ are the two eigenvalues of $\frac{H^0+H^{\mathrm{I}}}{\hbar}$, the \QCurve  $ \big(\psi_1(\mathbf{r}), U(t)=\ee^{-\iu \frac{(H^0+H^{\mathrm{I}})}{\hbar} t}, \deltau t \big)$ is in $\setOscillating$.
\end{theorem}
\begin{proof}
From~\eqref{eq:Fermi-coefficient-two-state}, the evolution of the state $\ket{\psi_{E_1}}$ can be described as
\begin{align}
    \ket{\psi_t}&=   \alpha_1(t) \ket{\psi_{E_1}} + \alpha_2(t) \ket{\psi_{E_2}}\,,
\end{align}
where $\alpha_1(t)=\cos^2 \theta \, \ee^{-\iu  \lambda_1 t}  + \sin^2 \theta \,  \ee^{-\iu  \lambda_2 t}$, $\alpha_2(t)=\sin 2 \theta \,  \left(\frac{\ee^{-\iu  \lambda_1 t} - \ee^{-\iu  \lambda_2 t}}{2} \right)$, the parameter $\theta$ is defined by~\eqref{eq:theta-definition},  the eigenvalues $\lambda_1,\lambda_2$ by~\eqref{eq:eigenvalues-two-state}, and all are functions of the six energy parameters defining $H^0$ and $H^1$.  We can then compute the position and momentum probability amplitudes  of the state,
$\psi(\mathbf{r},t)= \bra{\mathbf{r}}\ket{\psi_t}$ and $\phi(\mathbf{k},t)= \bra{\mathbf{k}}\ket{\psi_t}$, and then the resulting probability density functions are
\begin{align}
     \rho_{\mathrm{r}}(\mathbf{r},t)  &= |\alpha_1(t)\psi_1(\mathbf{r})+ \alpha_2(t) \psi_2(\mathbf{r})|^2
     \\
     &=  A_1(\mathbf{r})+A_2(\mathbf{r}) \sin^2  \frac{(\lambda_2-\lambda_1) \, t}{2} +A_3(\mathbf{r}) \sin  (\lambda_2-\lambda_1) \, t\,,
     \\
     \rho_{\mathrm{k}}(\mathbf{k},t) &= |\alpha_1(t)\phi_1(\mathbf{k})+ \alpha_2(t) \phi_2(\mathbf{k})|^2
     \\
     &=B_1(\mathbf{k})+B_2(\mathbf{k}) \sin^2  \frac{(\lambda_2-\lambda_1) \, t}{2}+ B_3(\mathbf{k}) \sin  (\lambda_2-\lambda_1)\, t\,,
\end{align}
where
\begin{align}
    A_1(\mathbf{r}) &=  |\psi_1(\mathbf{r})|^2 \,,
    \\
    A_2(\mathbf{r})&=\sin^2 2\theta\, \left ( |\psi_2(\mathbf{r})|^2+2\cos2 \theta \, \Re(\psi_1^*(\mathbf{r})\psi_2(\mathbf{r})) -|\psi_1(\mathbf{r})|^2\right)\,,
    \\
    A_3(\mathbf{r})&= \sin 2 \theta \, \Im (\psi_1^*(\mathbf{r})\psi_2(\mathbf{r}))\,,
    \\
     B_1(\mathbf{k}) &=  |\phi_1(\mathbf{k})|^2\,,
    \\
    B_2(\mathbf{k})&=\sin^2 2\theta\, \left ( |\phi_2(\mathbf{k})|^2+2\cos2 \theta \, \Re(\phi_1^*(\mathbf{k})\phi_2(\mathbf{k})) -|\psi_1(\mathbf{r})|^2\right)\,,
    \\
   B_3(\mathbf{k})&= \sin 2 \theta \, \Im (\phi_1^*(\mathbf{k})\phi_2(\mathbf{k}))\,.
\end{align}
The entropy will then vary in time according to the two densities $\rho_{\mathrm{r}}(\mathbf{r},t)$ and $\rho_{\mathrm{k}}(\mathbf{k},t)$.  Clearly, the entropy will return to the same value  at every period $T=\frac{\pi}{|\lambda_2-\lambda_1|}$, which is  a period where $  \sin  (\lambda_2-\lambda_1) \, t$ and  $\sin^2  \frac{(\lambda_2-\lambda_1) \, t}{2}$ simultaneously return to the same values. As the entropy is not constant, it must oscillate during this period. Therefore the \QCurve $ \big(\psi_1(\mathbf{r}), U(t)=\ee^{-\iu (H^0+H^{\mathrm{I}}) t}, \deltau t \big)$ is in $\setOscillating$.
\end{proof}
The theorem elucidates  that for the simple case of two states  the probability of transitions induces an entropy oscillation, without requiring   any  approximation.

We also have shown before that the method used to derive the coefficient $\alpha_2(t)$ of transitions to the second state can be expanded to multiple states. However, for many states,  the sum over all frequencies  $\lambda_k-\lambda_i$ may cancel the oscillations unless some frequencies dominate the sum, such as the transition to the ground state being much larger than other transitions. Thus, proving the entropy oscillation in the presence of multiple transitions may require the use of Fermi's golden rule approximations.

\subsection{Time Reflection}
\label{subsec:time-reflection}
The notation of this section is in the QFT setting, but it can be easily adapted to the quantum \mechanics setting.
Consider a neutral particle following a \QCurve in \setOscillating, such  that for a period of time $[0,t_c]$ it is in \setIncreasing and then it enters the decreasing interval at the critical time $t_c$. We consider the Hamiltonian to be time independent, with energy conservation over time translation, and investigate the discrete symmetries C and P and propose the time reflection to be augmented with time translation, say by $\deltau t$. We refer to this as time reflection,  because as we vary $t$ from $0$ to $\deltau t$, $-t+\deltau t$ varies as a  reflection from $\deltau t$ to $0$.
Time Reflection is described by the classical mapping $ t \mapsto -t + \deltau t$,  which is obviously  an involution.

 We  define the Time Reflection  quantum field
\begin{align}
\label{eq:time-reflection-T}
\Psi^{\mathrm{T}_{\deltau}}(\mathbf{r},-t+\deltau t)\eqdefA \mathscr{T} \Psi(\mathbf{r},t)= T \Psi^*(\mathbf{r},t)\,.
\end{align}
Time Reflection is an antilinear  and an antiunitary  transformation. It is an involution on the probability density $\rho(\mathbf{r},t)=\Psi^{\dagger}(\mathbf{r},t)\Psi(\mathbf{r},t)$.
Note that  in contrast to the case of time reversal, here  we have  $\Psi^{\mathrm{T}_{\deltau}}(\mathbf{r},t')=\mathscr{T} \Psi(\mathbf{r},-t'+\deltau t)$ and the entropies associated with $\Psi(\mathbf{r},t)$ and  $\Psi^{\mathrm{T}_{\deltau}}(\mathbf{r},t)$  are generally not equal. This means that an instantaneous time reflection transformation will cause entropy changes.

We next consider a composition of the  three transformation, Charge Conjugation, Parity Change, and Time Reflection.
\begin{definition}[$\Psi^{\CPT_{\deltau}}$]
  \label{def:cpt-deltau-quantum-field}
  We define the $\mathrm{CPT}_\deltau$ quantum field as
\begin{align}
\label{eq:CPT-delta-t}
    \Psi^{\mathrm{CPT}_{\deltau}}(-\mathbf{r},-t+{\deltau t})& \eqdefA \eta_{\mathrm{CPT}_{\deltau}}\,  CPT\,  \overline{\Psi}^{\tran}(\mathbf{r},t)=\eta_{\mathrm{CPT}_{\deltau}} \gamma^5 \, (\Psi^{\dagger})^{\tran}(\mathbf{r}, t)\,,
\end{align}
where $\eta_{\mathrm{CPT}_{\deltau}}$ is the product of the phases of each operation (and can be assumed to be 1),
 and $\gamma^5\eqdefA\iu \gamma^0\gamma^1\gamma^2\gamma^3$.
\end{definition}
The QFT invariance under CPT and energy conservation (to guarantee time translation invariance)  will guarantee $\CPT_{\deltau}$ invariance for the same QFT, that is,  if $\Psi(\mathbf{r},t)$ is a solution to a QFT then $\Psi^{\mathrm{CPT}_{\deltau }}(-\mathbf{r},-t+{\deltau t})$ is  also a solution to that QFT. CPT$_{\deltau}$  is an involution since
\begin{align}
   (\Psi^{\mathrm{CPT}_{\deltau}})^{\mathrm{CPT}_{\deltau}} (-(-\mathbf{r}),-(-t+{\deltau t})+{\deltau t}) &= \eta_{\mathrm{CPT}}\, \gamma^5 \left ( \left (\eta_{\mathrm{CPT}}\, \gamma^5 (\Psi^{\dagger})^{\tran})(\mathbf{r},t) \right) ^{\dagger}\right)^{\tran}
   \\
   &\Downarrow
   \\
  (\Psi^{\mathrm{CPT}_{\deltau}})^{\mathrm{CPT}_{\deltau}}(\mathbf{r},t) &=\eta_{\mathrm{CPT}_{\deltau}}\, \gamma^5 \eta_{\mathrm{CPT}_{\deltau}}^* \gamma^{5*} \, \Psi(\mathbf{r},t)
=\Psi(\mathbf{r},t)\,.
\label{eq:CPTdeltaInvolution}
\end{align}
\begin{definition}[$Q_{\CPT_{\deltau}}$]
  \label{def:cpt-transformation}
  We define $Q_{\CPT_{\deltau}}$, an instantaneous \QCurve transformation, to be $Q_{\CPT_{\deltau}}:\big(\psi(\mathbf{r},t_0), U(t), [t_0,t_1] \big)  \mapsto  \big(\psi^{\CPT_{\deltau}}(-\mathbf{r},t_0), U(t),  [t_0,t_1]  \big) $ with $\deltau t=t_0+t_1$.
\end{definition}
Note that $\psi^{\CPT_{\deltau}}(-\mathbf{r},t_0)\eqdefA \eta_{\mathrm{CPT}_{\deltau}} \gamma^5 \, (\Psi^{\dagger})^{\tran}(\mathbf{r}, -t_0+\deltau t)$  allows time reflections (time translations and time reversal), so that $\psi^{\CPT_{\deltau}}(-\mathbf{r},t_0)=\eta_{\CPT_{\deltau}}\, \gamma^5\, (\Psi^{\dagger})^{\tran}(\mathbf{r},t_1)$ for $\deltau t=t_1+t_0$, and so $t_1=-t_0+\deltau t$.

\begin{figure}
 \centering
 \includegraphics[scale=0.55]{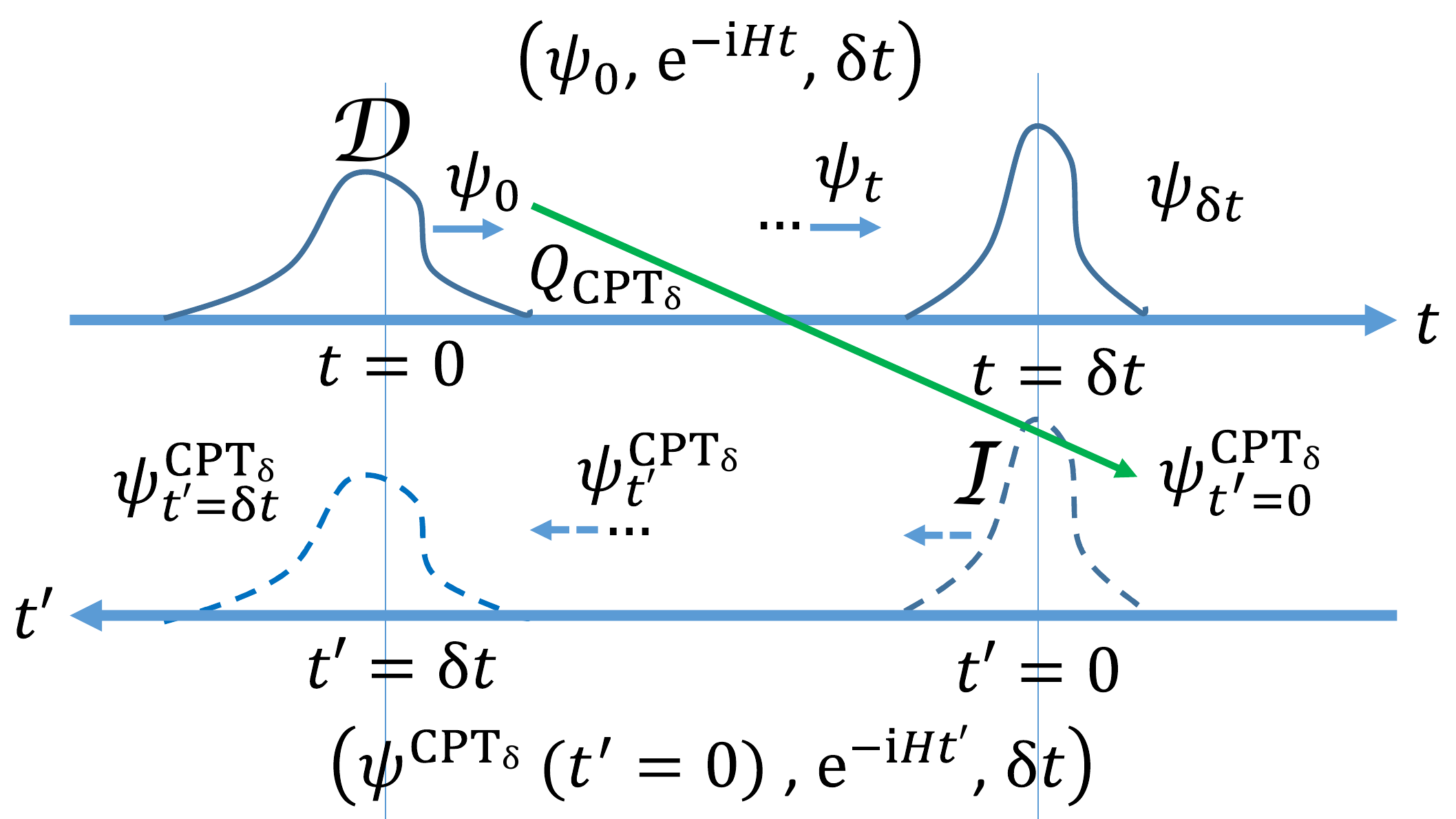}
 \caption{The scenario of Theorem~\ref{thm:CPT-QFT}. (i) Top axis $t$: a depiction of \QCurve $e_0=\big (\psi_0(\mathbf{r}), \eu^{-\iu H t}, [0,\deltau t] \big)$. (ii) Bottom axis $t' =\deltau t-t$: The antiparticle \QCurve  is created as  $e_1=Q_{{\CPT}_{\deltau}}(e_0)=\big (
     \psi^{{\CPT}_{\deltau}}(-\mathbf{r},t'=0)
     , \eu^{-\iu H t'}, [0,\deltau t] \big)$. The axis $t'$ shows the evolution as going forward in time $t'$.  Theorem~\ref{thm:CPT-QFT} can be visualized  by the evolution of state $\psi^{{\CPT}_{\deltau}}(-\mathbf{r},t')=\eta_{\mathrm{CPT}_{\deltau}} \gamma^5 \, (\Psi^{\dagger})^{\tran}(\mathbf{r}, \deltau t - t')$ as mirroring the evolution of state $\psi(\mathbf{r},t)$ with $t=t'$ evolving from $0$ to $\deltau t$. If  $e_0 \in \setDecreasing $, then  $e_1 \in \setIncreasing$.
}
 \label{fig:QCPT}
\end{figure}

\begin{theorem}
 \label{thm:CPT-QFT}
 Consider a $\mathrm{CPT}$ invariant quantum field theory (QFT) with energy conservation, such as Standard Model or Wightman axiomatic QFT. Let $e_0=\left (\psi(\mathbf{r},t_0), U(t), [t_0,t_1]\right)$ be a  \QCurve  solution to  such QFT.  Then, $e_{1} = Q_{\CPT_{\deltau}}(e_{0})$  is (i) a solution to such QFT, (ii) if $e_{1}$ is  in
 \setConstant,  \setDecreasing , \setOscillating, \setIncreasing then $e_{2}$ is respectively in \setConstant,  \setIncreasing , \setOscillating, \setDecreasing, making \setConstant,  \setIncreasing , \setOscillating, \setDecreasing reflections of \setConstant,  \setDecreasing , \setOscillating, \setIncreasing, respectively.
\end{theorem}

\begin{proof}
Consider a $t$-reflected time parameter
$t'=t_1 - (t - t_0)=\deltau t - t$, with $\deltau t=t_1+t_0$.
The \QCurve~$e_1$ describes  the evolution  $\psi^{\CPT_{\deltau}}(-\mathbf{r},t')$ (see  \eqref{eq:CPT-delta-t})  during  the period $[t_0,t_1]$, where $\deltau t= t_1+t_0$, so that $\psi^{\CPT_{\deltau}}(-\mathbf{r},t'=t_0)$ at $t'=\deltau t - t_1$.

To prove (i) we just use the assumption that  $e_0$ is a solution to a QFT that is $\mathrm{CPT}$ invariant and due to the energy conservation is alto time translation invariant and $e_1$ is the result of such transformations.

To prove (ii) we  apply $t'=\deltau t - t$ to $\psi^{\CPT_{\deltau}}(-\mathbf{r},\deltau t-t)$
 and we get
\begin{align}
   \psi^{\CPT_{\deltau}}(-\mathbf{r},t')& =CPT (\Psi^{\dagger})^{\tran}(\mathbf{r},t_1-(t'-t_0))\,.
     \label{eq:time-reversal-evolution}
\end{align}
In particular, for $t'=t_1 $,  $\psi^{\CPT_{\deltau}}(\mathbf{r}, t_1) =U(t_1-t_0)\psi^{\CPT_{\deltau}}(\mathbf{r}, t_0)$, and  one can reach the  transformed initial state of $e_0$ by starting with the transformed final state of $e_0$ and evolving it through the QFT for the time interval $t_1-t_0$. The evolution of $\psi(\mathbf{r},t)$ as $t$ evolves from $t_0$ to $t_1$, by Lemma~\ref{lemma:Entropy-invariants}, has the same entropies as $\psi^{\CPT_{\deltau}}(-\mathbf{r}, t_1-(t-t_0))$.  Then,  $e_1$  produces the time evolution states $\psi^{\CPT}(-\mathbf{r}, t_1-(t-t_0))$, which traverse the same path  as $\psi^{\CPT}(-\mathbf{r}, t)$, but in the opposite time directions.

For a \QCurve in \setIncreasing (or in \setDecreasing) the  $\entropyS_t$ increasing (or decreasing) from $t_0$ to $t_1$. By reversing time and starting from the end $\entropyS_{t_1}$ to its beginning $\entropyS_{t_0}$, $\entropyS_t$ will then be decreasing (or increasing), and thus in \setDecreasing (or in \setIncreasing).

For a \QCurve in \setConstant the entropy is constant from $t_0$ to $t_1$. By reversing the time from $t_1$ to $t_0$, the entropy will also  be constant, thus in \setConstant.

For a \QCurve in \setOscillating the entropy is oscillating from $t_0$ to $t_1$. By reversing the time from $t_1$ to $t_0$, the entropy will also  be oscillating, thus in \setOscillating.

Time  reversal for the period $[t_0,t_1]$ starting from $t_1$ and ending at $t_0$ is equivalent to  time reflection by an amount $\deltau t= t_1+t_0$. Thus, we conclude that \setConstant,  \setIncreasing , \setOscillating, \setDecreasing are reflections of \setConstant,  \setDecreasing , \setOscillating, \setIncreasing, respectively.
 \end{proof}

A visualization of the ideas of this theorem is given in Figure~\ref{fig:QCPT}.

\begin{theorem}
    \label{thm:bijection}
  Let the four blocks for the partition of $\setAll$ be $\setAll_{i}$, $i = 1, 2, 3, 4$. Let $\setAll_{k}$ be the reflection of  $\setAll_{j}$.
Then  the restriction of $Q_{\CPT}$ to $\setAll_{j}$,
  $Q_{\CPT}\big|_{\setAll_{j}}$, is a bijection between $\setAll_{j}$ and $\setAll_{k}$.
\end{theorem}

\begin{proof}
 By Theorem~\ref{thm:CPT-QFT} $Q_{\CPT}$ is an injection of the set into its reflection.
 As $Q_{\CPT}$ is an involution (see \eqref{eq:CPTdeltaInvolution}), the union of the four ranges
$\bigcup_{1}^{4} Q_{\CPT}(\setAll_{i})$ is $\setEvolutions$. As  the ranges are disjoint, each
 $Q_{\CPT}(\setAll_{i})$,
 $i=1,2,3,4$ is surjective and therefore also a bijection. \qedhere
\end{proof}

\subsection{Two-Particle System}
\label{subsec:two-particle-system}

The notation of this section is in the quantum \mechanics setting, but it can be easily adapted to the QFT setting.
Consider a  two-fermions or a two-bosons system.  The state is then a ray element of the product of two Hilbert spaces and the evolution of a two-particle system made of  fermions (f) or of bosons (b) is given by
 \begin{align}
  \label{eq:two-particle-state}
   \ket{\psi_t^{\fer,\bos}}&=\frac{1}{\sqrt{C_t}}  \left (\ket{\psi^1_t}\ket{\psi^2_t}\mp \ket{\psi^2_t}\ket{\psi^1_t}
  \right )\,,
  \end{align}
 where $C_t$ is the normalization constant that may evolve over time and the signs ``$\mp$'' represent the fermions (``$-$'') and bosons (``+''). When the two states, $\ket{\psi^1_t}$ and $ \ket{\psi^2_t}$,  are orthogonal to each other,  $C_t=2$.
 Projecting on $\bra{\mathbf {r}_1}\bra{\mathbf {r}_2}$ and on $\bra{\mathbf {p}_1}\bra{\mathbf {p}_2}$ we get
 \begin{align}
     \psi^{\fer,\bos}(\mathbf {r}_1,\mathbf {r}_2 ,t)&=\frac{1}{\sqrt{C_t}}
     \left ( \psi_1(\mathbf {r}_1,t) \psi_2(\mathbf {r}_2,t)\mp \psi_1(\mathbf {r}_2,t) \psi_2(\mathbf {r}_1,t)  \right )\,,
      \\
     \psi^{\fer,\bos}(\mathbf {p}_1,\mathbf {p}_2 ,t)&=\frac{1}{\sqrt{C_t}}
     \left ( \phi_1(\mathbf {p}_1,t) \phi_2(\mathbf {p}_2,t)\mp \phi_1(\mathbf {p}_2,t) \phi_2(\mathbf {p}_1,t)  \right )\,.
 \end{align}
 The probability density functions are then
  \begin{align}
  \label{eq:density-two-particles}
     \rho_{\mathrm{r}}^{\fer,\bos}(\mathbf {r}_1,\mathbf {r}_2 ,t)& =|\psi^{\fer,\bos}(\mathbf {r}_1,\mathbf {r}_2 ,t)|^2
     \\
     &=\frac{1}{C_t}\left [ \rho_{\mathrm{r}}^1(\mathbf {r}_1,t) \rho_{\mathrm{r}}^2(\mathbf {r}_2 ,t)+\rho_{\mathrm{r}}^1(\mathbf {r}_2,t) \rho_{\mathrm{r}}^2(\mathbf {r}_1 ,t) \right .
     \\
     & \qquad \left . \mp 2 \Re \left (
     \psi_1(\mathbf {r}_1,t) \psi_2^*(\mathbf {r}_1,t) \psi_2(\mathbf {r}_2,t)
     \psi_1^*(\mathbf {r}_2,t)  \right) \right ]\,,
 \\
  \rho_{\mathrm{p}}^{\fer,\bos}(\mathbf {p}_1,\mathbf {p}_2 ,t)& =|\phi^{\fer,\bos}(\mathbf {p}_1,\mathbf {p}_2 ,t)|^2
     \\
     &=\frac{1}{C_t}\left [ \rho_{\mathrm{p}}^1(\mathbf {p}_1,t) \rho_{\mathrm{p}}^2(\mathbf {p}_2 ,t)+\rho_{\mathrm{p}}^1(\mathbf {p}_2,t) \rho_{\mathrm{p}}^2(\mathbf {p}_1 ,t) \right .
     \\
     & \qquad \left . \mp  2 \Re \left (
     \phi_1(\mathbf {p}_1,t) \phi_2^*(\mathbf {p}_1,t) \phi_2(\mathbf {p}_2,t)
     \phi_1^*(\mathbf {p}_2,t)  \right) \right ]\,.
 \end{align}
 The interference terms (the third terms) characterize the entanglement of the particles.
 The entropy of the two-particle system, from \eqref{eq:entropy-many-particles},  is then
  \begin{align}
     \label{eq:entropy-two-particles}
     S^{\fer,\bos}\left (\ket{\psi^1_t},\ket{\psi^2_t}\right )      &= - \int \! \diff^3 \mathbf {r}_1 \int \diff^3 \mathbf {r}_2\,   \rho_{\mathrm{r}}^{\fer,\bos}(\mathbf {r}_1,\mathbf {r}_2 ,t)
     \ln \rho_{\mathrm{r}}^{\fer,\bos}(\mathbf {r}_1,\mathbf {r}_2 ,t)
     \\
     & \quad - \int \! \diff^3 \mathbf {p}_1 \int \diff^3 \mathbf {p}_2\,   \rho_{\mathrm{p}}^{\fer,\bos}(\mathbf {p}_1,\mathbf {p}_2 ,t)
     \ln \rho_{\mathrm{p}}^{\fer,\bos}(\mathbf {p}_1,\mathbf {p}_2 ,t)
     \\
     &\quad - 6 \ln \hbar\, .
  \end{align}
The collision of two particles is of interest.  Let us  analyze  two coherent wave packet  solutions moving towards each other along the $x$-axis, and  with momenta $p_1$ and $-p_1$, starting at centers $c_1$ and $c_2$  and with the same position variance $\sigma^2$. They can be represented in position and momentum space as
\begin{align}
\Psi_{1}(x,t) &=\frac{\eu^{-\iu p_1 v_p(p_1) \, t}}{Z_1}\,  \normalx{x}{c_1 +v_g(p_1) \, t}{\sigma^2  +\iu \, t\, {\cal H}(p_1) }\eu^{\iu p_1 x} \,,
 \\
\Psi_{2}(x,t) &=\frac{\eu^{-\iu p_1 v_p(p_1) \, t}}{Z_1}\,\,  \normalx{x}{c_2 -v_g(p_1) \, t}{\sigma^2 +\iu \, t\, {\cal H}(-p_1) } \eu^{-\iu p_1 x}\, ,
\\
 \Phi_1(k, t) & =   \frac{\eu^{-\iu t\, v_p(p_1) p_1} }{Z_{p_1}} \, \, \normalx{k}{p_1}{(\sigma^2  +\iu \, t\, {\cal H}(p_1))^{-1} }\, \eu^{\iu (k-p_1)\left (c_1+v_g(p_1) \, t\right)}\,,
\\
 \Phi_2(k, t) & =   \frac{\eu^{-\iu t\, v_p(p_1) p_1} }{Z_{p_1}} \, \, \normalx{k}{-p_1}{(\sigma^2  +\iu \, t\, {\cal H}(-p_1))^{-1} }\, \eu^{\iu (k+p_1) \left (c_2-v_g(p_1) \, t\right)}\, .
\label{eq:two-coherent-states}
\end{align}

\begin{figure}
 \centering
\includegraphics[scale=0.3]{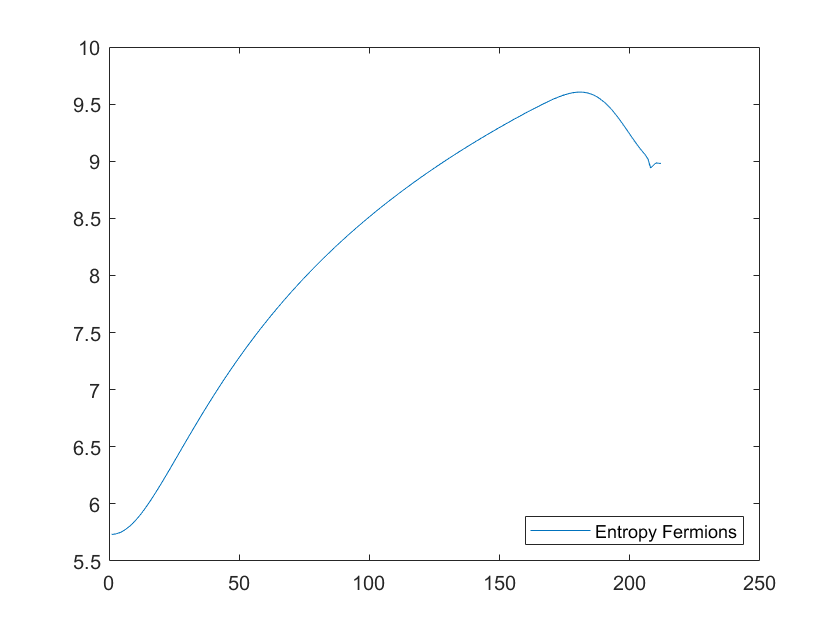}\includegraphics[scale=0.25]{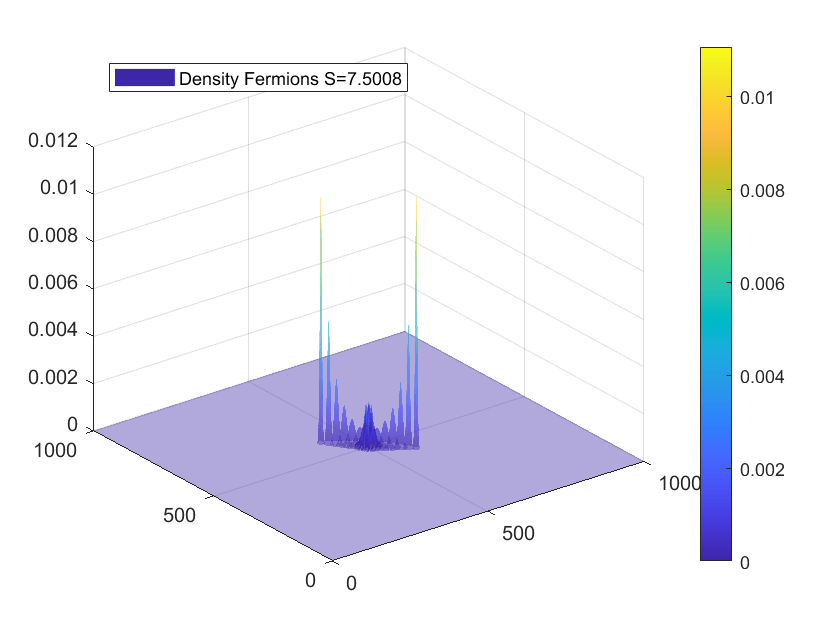}
 \\
 \quad  (a) $p_1=1, \frac{\hbar}{m}=1$ : Entropy versus time   \qquad  \qquad  $\rho_{x}^{\fer}(x_1,x_2 ,t)$  overlayed over time
 \\
 \includegraphics[scale=0.3]{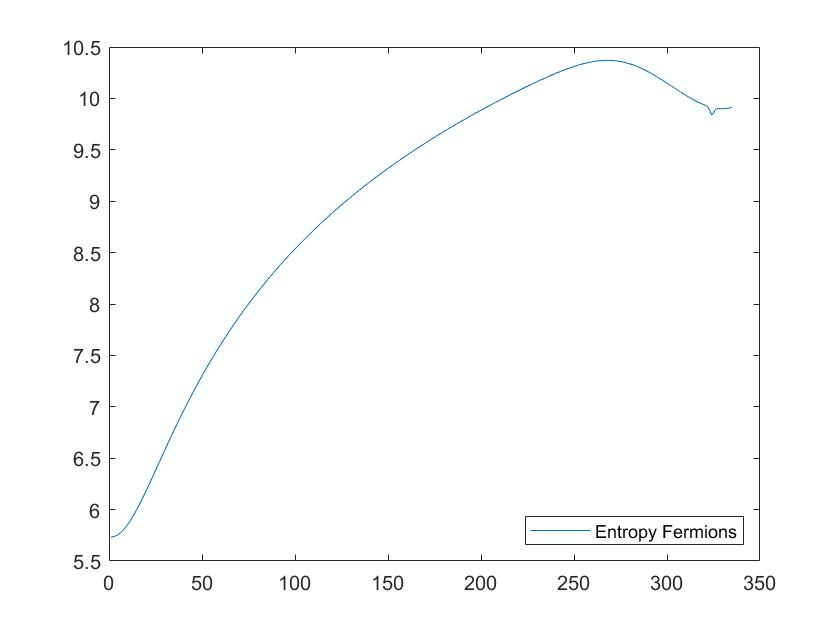}\includegraphics[scale=0.25]{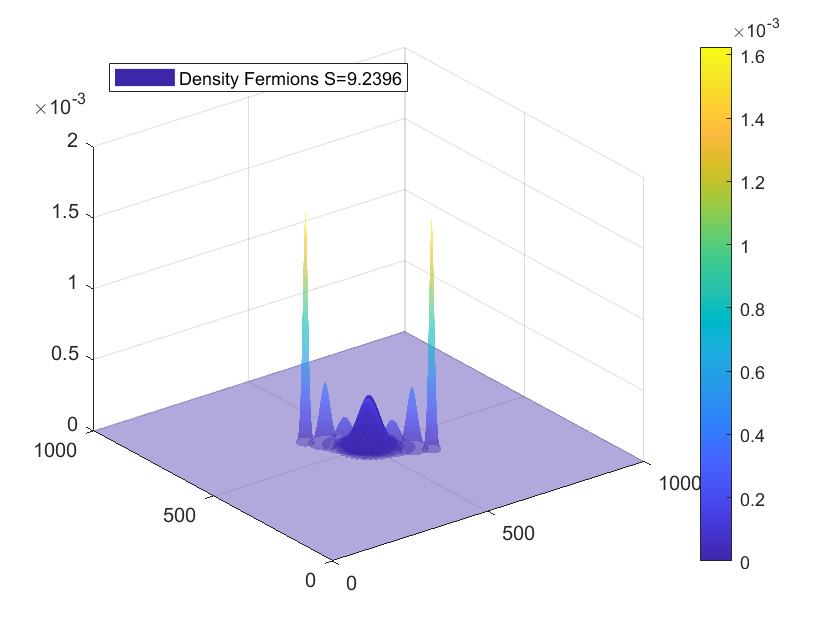}
 \\
 \quad (b) $p_1=1, \frac{\hbar}{m}=0.5$ : Entropy versus time  \qquad  \qquad  $\rho_{x}^{\fer}(x_1,x_2 ,t)$  overlayed over time
 \\
 \includegraphics[scale=0.3]{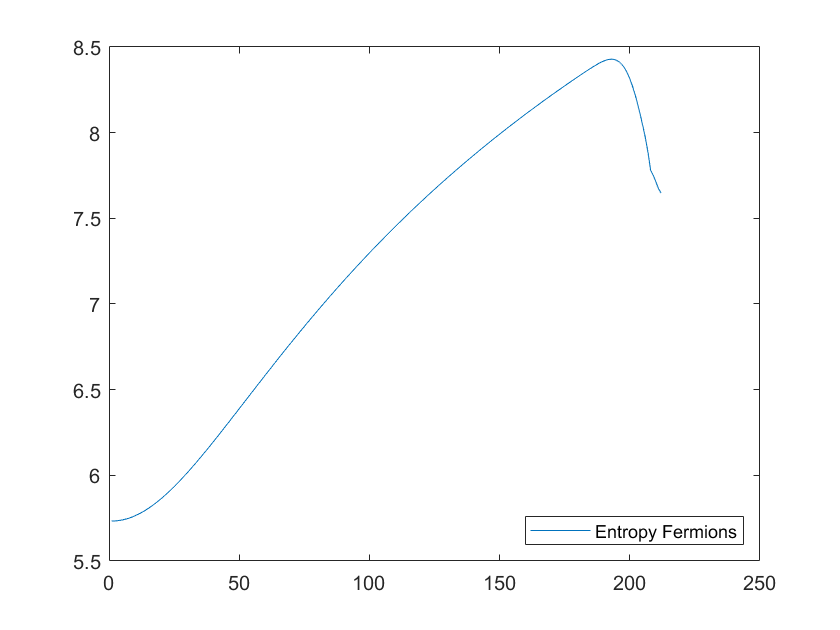}\includegraphics[scale=0.25]{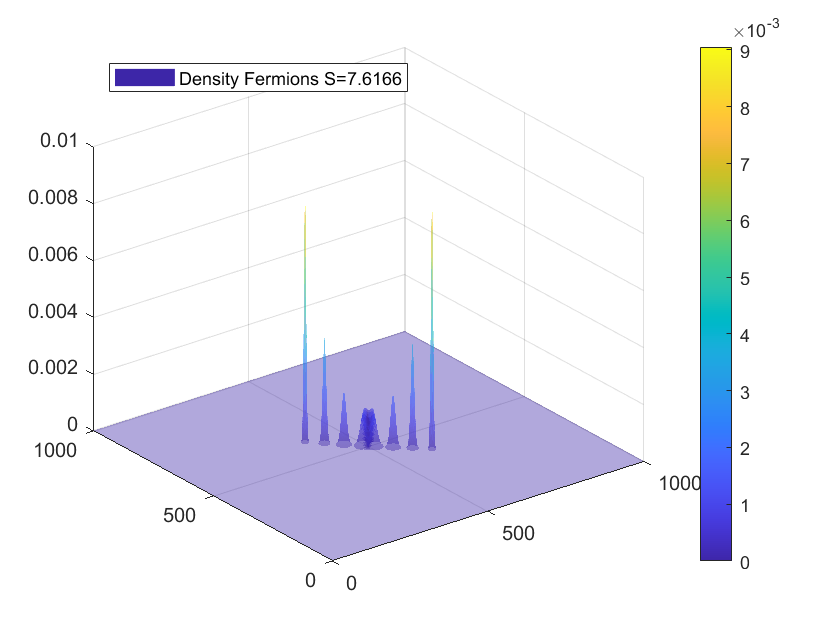}
 \\
  \quad  (c)  $p_1=2, \frac{\hbar}{m}=0.5$ : Entropy versus time  \qquad  \qquad   $\rho_{x}^{\fer}(x_1,x_2 ,t)$  {overlayed over time}
 \caption{{Collision of two fermions described by probability densities \eqref{eq:density-two-particles}, individual probability amplitudes \eqref{eq:two-coherent-states}, and  common parameters  $c_2=-c_1=150$, with speed of light $c=1$, a grid of $1\,000$ points for $x_1, x_2,k_1,k_2$. The three graphs reflect changes in $p_1$ and $\frac{\hbar}{m}$. The left column graphs show entropy values vs time. The right column graphs show snapshots of the density function at intervals of 20 time units,  overlaid on  a single plot. The $z$-axis represents the density values while the $x$-$y$ axes represent $x_1$-$x_2$ values. As the particles approach each other, the individual densities disperse, the maximum values are reduced, and the entropy increases. Only when the particles  are close to each other, the interference effect reduces the total entropy.}  }
 \label{fig:two-particle-collision-entropy}
\end{figure}

Figure~\ref{fig:two-particle-collision-entropy} show plots of the entropy versus time when such two particles in coherent states are moving towards each other. Next to each entropy plot vs time, there is a graph that shows snapshots of the density function $\rho_{x}^{\fer}(x_1,x_2 ,t)$ at intervals of 20 time units apart,  overlaid on a single plot.   When  the two particles  are far apart, the total entropy of the system is close to the sum of the two individual  entropies, and each one is increasing over time. As they come closer to each other, the overlap of the individual probabilities increases, with the resulting interference tending to lower the total entropy.
Such a competition between the two factors, the increase of the entropy of the individual particles and the decrease of entropy due to interference (entanglement), will result in an oscillation and the decrease in total entropy when the two particles come close together. The three graphs reflect changes in the parameters $p_1$ and $\frac{\hbar}{m}$, with $v_g(p_1)$ derived in \eqref{eq:Fourier-group-velocity} and ${\cal H}(p_1)$ derived in \eqref{eq:Fourier-group-Hessian}.  Comparing Figures~\ref{fig:two-particle-collision-entropy}(a) and (b): the larger is the particle mass (the smaller is $\frac{\hbar}{m}$) the slower is the speed $v_g(p_1)$, the slower is the collision ($x$-axis values), and the later  entropy  starts  decreasing. Thus, light, fast particles will enter the oscillation period sooner. Comparing Figures~\ref{fig:two-particle-collision-entropy}(b) and (c): The larger is the speed $v_g(p_1)$, the sooner is the collision ($x$-axis values), and  the sooner the entropy starts decreasing. Thus, it is possible that colliding particles will be in  \setOscillating, even if each one alone has a dispersion property that increases their entropy.

\section{The Entropy Law}
\label{sec:entropy-law}
In classical statistical \mechanics, the entropy provides a time arrow through the second law of thermodynamics \cite{clausius1867mechanical}.  We have  shown that due to the dispersion property of the fermions Hamiltonian some states in quantum \mechanics, such as  coherent states,  already obey such a law.  However, current quantum theory is time reversible.

To provide a time arrow for all particles, we propose  a law analogous to the second law of thermodynamics

\begin{postulate}[The Entropy Law]
 \label{postulate:1}
 The entropy of a physical evolution of a quantum system is an increasing function of time.
\end{postulate}

This law can be incorporated into the path integral formalism in its kernel propagation formulation:
\begin{align}
    \psi(x,t) = \int \diff y \int_{ {\cal P}_{y,x} \, , \, \entropyS^{\uparrow}} \diff {\cal P}_{y,x} \eu^{\iu {\cal A}_{t_0}^t({\cal P}_{y,x})}\, \psi(y,t_0)\,,
\end{align}
where ${\cal P}_{y,x}$ is a path from position $x$ to position $y$, ${\cal A}_{t_0}^t({\cal P}_{y,x})$ is the action associated with this path from $t_0$ to $t$, and $\entropyS^{\uparrow}$ represents the constraint that entropy must increase throughout this path, that is,  the integral is restricted  to only those paths in which the entropy increases. An alternative formulation would consider a standard QFT
Lagrangian and add to it a Lagrange multiplier term with the time varying entropy  being non-negative, obtaining  $ \mathscr{ L}(t) +\lambda  \partial_t \entropyS$.

This law aims to treat quantum \physics an irreversible  statistical theory, and thus to create the time arrow as the arrow of information loss (increase in entropy). Reversibility occurs for stationary states as the entropy $\entropyS$ remains constant.
Concerning other states,
\QCurves in the set \setDecreasing would  be ruled out and \QCurves in set \setOscillating would not be allowed to complete their evolution. The entropy law can be used to establish causes for physical phenomena for which such causes have not been known. We will  focus on \QCurves in \setOscillating and to what happens to particles when they can not complete the oscillations.

{In QFT, the operators $\Psi(\mathbf{r},t)$ and $\Phi(\mathbf{k},t)$ are expressed in terms of the creation and the annihilation operator and relate to each other via the Fourier transform. These operators act on a state $\ket{n}$ with $n$ particles producing new states with a  creation  or destruction of  a particle or anti-particles.   These operators are thought to act instantaneously and must satisfy conservation laws, e.g. in the hydrogen atom,  a photon created upon the jump of an electron from an excited state to the ground state must have the energy equal to the reduction of the electron's energy.  }

The cause for such instantaneous transformations (creation and annihilation) is  speculated to follow from the entropy law, as we examine and discuss next. The instant in the oscillation period when the entropy starts to decrease is when the  transitions of electrons in atoms occur with photons being emitted or absorbed, oscillations and decay of neutral particles happen, and when collisions of particles lead to their transformations into other particles.
For all such transformations  the conservation laws must be preserved, but it is the entropy law that  triggers  those transformations.

We now consider some scenarios of interest. What happens at the instance when particles are transformed into other particles or new particles are created, such as emissions of photons when electrons transition from excited states to the ground state in an atom?
These are instantaneous events, so could the entropy be decreased in such an event?
After all, the creation of new particles may require an increase in information, and if so, is there a source of information in the vacuum for reducing the entropy?
We do not have answers to such questions. However, if such instantaneous transitions with entropy decreaing are possible, then we exploit in Section~\ref{subsec:kaons} some consequences accounting for particle and antiparticle transformations.

\subsection{The Stability of the Ground State of the Hydrogen Atom  }
\label{subsec:stab-hydr-ground}

This section we start with the QED formulation for the hydrogen atom using a QFT setting. However, the solutions for the states of the electron are obtained from the \Schroedinger approximation, that is, they are solutions using the quantum \mechanics representation.

The QED Hamiltonian  for the hydrogen atom is
\begin{align}
    H(p,r,q) & = \sum_{i=1}^3 \frac{\left( p^i-\frac{\ee}{c} A^i(q)\right )^2}{2 m} - \frac{\ee^2}{r}+\sum_{\lambda=1}^2  \planckbar \omega_{q} \, a^{\dagger }_{\lambda} (q)\, a_{\lambda}(q)\,,
\end{align}
where   photon's helicity  $\lambda$ is $1$ or $2$,  $\omega_q=|q|c$,  the creation and the annihilation operators  of photons satisfy
 $[a_{\lambda}(p), a^{\dagger }_{\lambda'} (q)]=\delta_{\lambda, \lambda'}\delta(p-q)$, and $A= (A^1, A^2, A^3)$.  The electromagnetic vector potential is
 \begin{align}
  \tilde A^i(q) & = \sqrt{2\piu  \planckbar c^2}\sum_{\lambda=1}^2 \frac{1}{\sqrt{\omega_q}} \left ( \epsilon^i_{\lambda}(q)\, a_{\lambda}(q)  + \epsilon^{*\iu }_{\lambda}(q) \, a^{\dagger}_{\lambda} (q)\right)\,,
 \end{align}
 and in  the Coulomb Gauge ($\nabla \cdot  A=0$), for $q=|q| ( \sin \theta_q \cos \phi_q, \sin \theta_q \sin  \phi_q, \cos \theta_q)$, the polarizations satisfy  $\epsilon_{1}(q) = ( \cos \theta_q \cos \phi_q, \cos \theta_q \sin  \phi_q, \sin \theta_q)$ and $\epsilon_{2}(q)  = (-\sin \phi_q, \cos \phi_q, 0)$.

 The state of the atom can be described by $\ket{\electron}\ket{\photon}=\ket{n,l,m}\ket{q,\lambda}$,
where  $n,l,m$ represent quantum numbers of the electron and $q,\lambda$ represent  the momentum and helicity of the photon.
In the absence of radiation or emission from  the atom, the photon state is the vacuum  $\ket{0}$. We next consider the  Lyman-alpha transition, $\ket{n=2,l=1,m=0}\ket{0} \, \rightarrow \, \ket{n=1,l-0,m=0}\ket{q,\lambda}$ with the emission of a photon with wavelength
 $\lambda\approx \SI{121.567d-9}{\metre}$.

 We first evaluate the electron's entropy at both states $\ket{n=2,l=1,m=0},\ket{n=1,l-0,m=0}$. For simplicity, we can consider the \Schroedinger approximation to describe the electron state with the energy change in this transition of $\Updelta E_{n=2\rightarrow n=1}\approx-\left(\frac{1}{2^2} - 1\right) \times \SI{13.6}{\electronvolt} = \SI{10.2}{\electronvolt}$.
\begin{enumerate}
    \item[(i)]The position probability amplitudes described in \cite{bransden2003physics} and the associated entropies are
    \begin{align}
        \psi_{2, 1, 0}(\rho,\theta,\phi) &= \frac{1}{\sqrt{32\piu}} \left(\frac{1}{a_{0}}\right)^{\frac{3}{2}}\: \rho \ee^{-\frac{\rho}{2}} \cos(\theta)\: \rightarrow \,  \entropyS_{\mathrm{r}}(\psi_{2, 1, 0})\approx \num{6.120}+\ln \piu+3 \ln a_0\,,
        \\
        \psi_{1, 0,0}(\rho,\theta,\phi) &= \frac{1}{\sqrt{\piu}} \left(\frac{1}{a_{0}}\right)^{\frac{3}{2}}\: \ee^{-\rho} \:  \rightarrow \,  \entropyS_{\mathrm{r}}(\psi_{1, 0,0})\approx\num{3.000}+\ln \piu+3 \ln a_0 \,,
    \end{align}
     where $a_0 \approx \SI{5.292d-11}{\metre}$ is the Bohr radius and $\rho={r}/{a_0}$.
    \item[(ii)] The momentum probability amplitudes described in \cite{bransden2003physics} and the associated entropies are
     \begin{align}
        \Phi_{2, 1, 0}(p, \theta_p, \phi_p)  &= \sqrt{\frac{128^2}{2\piu p_0^3}} \, \frac{p}{p_0} \, \left (1+\left(2\frac{p}{p_0}\right)^2\right )^{-3} \,\cos (\theta_p) \,,\\
        & \hspace*{-2em} \rightarrow \hspace{0.75em} \entropyS_p(\Phi_{2, 1, 0})\approx \num{0.042}+3 \ln p_0\,,
        \\
        \Phi_{1, 0, 0}(p, \theta_p, \phi_p) &= \sqrt{\frac{32}{\piu \, p_0^3}} \left ( 1+ \left(\frac{p}{p_0}\right)^2\right )^{-2} \,, \\
        & \hspace*{-2em}  \rightarrow \hspace{0.75em}  \entropyS_p(\Phi_{1, 0, 0})\approx \num{2.422}+3 \ln p_0\,  ,
    \end{align}
     where $p_0={\si{\planckbar}}/{a_0}$.
     \item[(iii)] Therefore, $\Updelta \entropyS_{2,1,0\rightarrow 1,0,0}=\entropyS_{\mathrm{r}}(\psi_{1, 0, 0})+ \entropyS_p(\Phi_{1, 0, 0}) -\entropyS_{\mathrm{r}}(\psi_{2, 1, 0})- \entropyS_p(\Phi_{2, 1, 0}) \approx \num{-0.740}\,.$
\end{enumerate}

Thus, the entropy of the electron is actually reduced by {approximately} $\num{0.740}$ when it moves $\ket{n=2,l=1,m=0} \, \rightarrow \, \ket{n=1,l-0,m=0}$.

The momentum wave function  derived  by \cite{lombardi2020hydrogen} is given by  \begin{align}
  \label{eq:9}
 \Phi'_{210}(p, \theta_p, \phi_p) &= \frac{2 (-1+\iu)}{\pi^2}\sqrt{2}\frac{1}{(i-2p)^3} \frac{(\theta_p \cos\theta_p-\sin(\theta_p))}{\theta_p^2} \frac{(\eu^{2 \iu \pi \phi_p}-1)}{\phi_p}\,,\\
\Phi'_{100}(p, \theta_p, \phi_p)     &= \frac{-1+\iu}{2\pi^2}\sqrt{2}\frac{1}{(i-p)^2} \frac{\sin(\theta_p)}{\theta_p} \frac{(\eu^{2 \iu \pi \phi_p}-1)}{\phi_p}\,,
\end{align}
where  a point transformation from the Cartesian coordinate system to the spherical coordinate system using the conjugate momentum operator \eqref{eq:conjugate-momentum-DeWitt} was considered. This yields the entropies $\entropyS'_p(\phi'_{210}) \approx 0.556$ and $\entropyS'_p(\phi'_{100})\approx 2.667$ respectively. These values show a change in the entropy of $\Updelta \entropyS'_{2,1,0\rightarrow 1,0,0}=\entropyS_{\mathrm{r}}(\psi_{1, 0, 0})+ \entropyS'_{\mathrm{p}}(\Phi'_{1, 0, 0}) -\entropyS_{\mathrm{r}}(\psi_{2, 1, 0})- \entropyS'_{\mathrm{p}}(\Phi'_{2, 1, 0}) \approx \num{-1.009}$.

Thus the entropy of the electron, according to the  momentum wave function  derived  by \cite{lombardi2020hydrogen}, is actually reduced by {approximately} $\num{1.009}$ when it moves $\ket{n=2,l=1,m=0} \, \rightarrow \, \ket{n=1,l-0,m=0}$.

Note that the expected value of the kinetic energy,  ${p^2}/{2\mu}$, for the ground state is larger than for the excited states
(scaling by a factor ${1}/{n^2}$ \cite{bransden2003physics}), which is consistent with the momentum entropy being larger for the ground state. The expected value of $r^2$ (the square of the distance from the center) increases with $n^4$ and also depends on $l$ \cite{bransden2003physics}, which is consistent with the position entropy being larger for the excited state.

We next evaluate the photon's entropy.  Due to energy conservation, the  energy must satisfy $|q| c\approx \SI{10.2}{\electronvolt}$, where $c$ is the speed of light. The
uncertainty in the value of $|q| \approx \num{10.2}/c\, \si{\electronvolt} $ is very small,
following from the uncertainty in the energy change. The main uncertainty for the photon is in specifying the direction of emission.   The  electron in the initial state $\ket{n=2,l=1,m=0}$ has  total angular momentum $l=1$, with no angular momentum ($m=0$) along $z$. In the final state $\ket{n=1,l=0,m=0}$, the total angular momentum of the electron is zero with no angular momentum ($m=0$) along $z$.  The spin $1$ of the photon conserves the total angular momentum of the system and it is along the motion of the photon.  Thus, the photon must be moving perpendicularly to the $z$ axis, that is, $\theta_q=\frac{\piu}{2}$ and so the polarization vectors  must be $   \epsilon_{1}(q) = ( 0, 0, 1) $ and $\epsilon_{2}(q)  = (-\sin \phi_q, \cos \phi_q, 0)$. The uncertainty is in the angle $\phi_q$, completely unknown,  therefore giving a momentum entropy $\entropyS_q=\ln 2\piu\approx \num{1.838}$. Similarly, the freedom in the momentum direction $\phi_q$ and the lack of freedom along $\theta_q=\frac{\piu}{2}$ leads to freedom in the photon direction $\phi$, and the $z$-angle must be $\theta =\frac{\piu}{2}$, while its speed is constant $c$. Therefore, the radius $r=|\mathbf{r}|$ is also always fully determined.  Thus the position entropy is also $\entropyS_{\mathrm{r}}=\ln 2\piu=\num{1.838}$.
The entropy is larger from the small uncertainties in the other variables, and by Proposition~\ref{proposition:lower-bound-S}, the photon's entropy should be larger.

Therefore, indeed, the entropy increases, because
\begin{align}
     \Updelta \entropyS_{\ket{n=2,l=1,m=0}\ket{0} \, \rightarrow \, \ket{n=1,l-0,m=0}\ket{q,\lambda}} &>  2\times \num{1.838} - \num{0.740}
     =\num{2.936}\, .
\end{align}
The increase in entropy also follows from the momentum wave function derived by  \cite{lombardi2020hydrogen}, where the entropy  increases by more than $2.667$.

According to QED, and due to photon fluctuations of the vacuum, the state of the electron in an excited state is in a superposition with the ground state, and from Theorem~\ref{theorem:oscillations} the entropy would  decrease in a time interval longer than ${\piu}/{|\omega_{n=2,l=1,m=0}-\omega_{n=1, l=0, m=0}|}$. Thus, according to the proposed entropy law, the electron must jump to the ground state and emit a photon.

Consider now  an apparent time-reversing situation, where a machine sends photons to strike a hydrogen atom with the electron in the ground state, where the  photons have energy $E_{\gamma}=\hbar |\omega_{n=2,l=1,m=0}-\omega_{n=1, l=0, m=0}|$. For the atom to absorb such a photon, the photon must have  followed a precise direction  towards the  atom, and a very small uncertainty in the direction implies low entropy of the photon. Once the atom absorbs the photon,  the electron in the ground state has the energy required to jump to the excited state. Note that the entropy law  is again verified because, as we showed above, the entropy of the excited state is larger than the entropy of the ground state and the photon entropy must have been a very small addition to the ground state entropy. Thus this transition to an excited state, which does occur,  also satisfies the entropy law.

\subsection{Neutral Particle Oscillations and Decay}
\label{subsec:kaons}

We next speculate on the mechanism for the oscillations and decay of a neutral K meson (kaon  $\mathrm{K}^0$, containing a down quark and a strange antiquark)  \cite{PhysRevLett.13.138}.  We speculate that  the \QCurve of $\mathrm{K}^0$, $e_0=(\psi_0(\mathbf{r}), U(t), \frac{2\piu}{\Updelta w})$, is in $\setOscillating$. Then, a  $\mathrm{K}^0$ particle in state $\psi_0(\mathbf{r})$ evolves in time and enters a decreasing period
at $T=\frac{\piu}{\Updelta w}$ when the  the remaining segment of \QCurve $e_T=(\psi_T(\mathbf{r}), U(t), \frac{\piu}{\Updelta w})$ is in \setDecreasing. In order  to prevent such a decrease, an instantaneous transformation takes place, when quarks exchange bosons to transform $\mathrm{K}^0 \mapsto \bar{\mathrm{K}}^0$ (down antiquark and a strange quark), to create an antiparticle \QCurve $e_1$ in \setIncreasing.

One possible candidate to describe this instantaneous transformation is $Q_{\CPT_{\deltau}}$ of Theorem~\ref{thm:CPT-QFT} that brings a particle's \QCurve $e_T$ to its new antiparticle \QCurve $e_1=Q_{\CPT_{\deltau}}(e_T)$, such that \QCurve $e_1 \in \setIncreasing$, without a decrease of entropy.  There is however one observation to be made: the entropy of $\psi(\mathbf{r},T)$ would be higher then the entropy of the initial state of \QCurve $e_1=Q_{\CPT_{\deltau}}(e_T)$. We wonder whether for the case of kaons, given that such transformation is instantaneous,  the entropy could instantaneously decrease?

We speculate that this happens to $\mathrm{K}^0$ entering $\setDecreasing$, during oscillations of the quarks,  then transforming into its antiparticle $\bar{\mathrm{K}}^0$. Then after a  period of ${\piu}/{\Updelta w}$,  $\bar{\mathrm{K}}^0$ would again enter $\setDecreasing$ when the entropy would start to decrease, and then a transformation back to $\mathrm{K}^0$  or a decay to mesons $\piu$ would occur,   where  we also speculate that the entropy would increase.

We speculate that free neutrinos when in a superposition of states with different masses are in states in  $\setOscillating$ and can exist in such a superposition only during an time interval when the entropy is increasing. If they move along the momentum direction corresponding to the eigenvalue $\lambda_1^D$ in \eqref{eigenvalues-Hamiltonian}, then the superposition will transition to the heavier mass state, corresponding to a larger dispersion. Otherwise, in the other two orthogonal directions the superposition will transition to the lighter mass state.  However, if neutrinos are Majorana particles, as speculated \cite{albert2014search}, then  a transformation $Q_{\CPT}: \ket{\psi} \mapsto \ket{\phi}=\eu^{\iu H \frac{\piu}{\Updelta w} }\ket{\psi^*}$ could be triggered before the entropy start decreasing. Such a process would be similar to the one described for kaons above.

\subsection{Particle Collisions and Oscillations}
\label{subsec:collisions}

We simulated two colliding fermions, each described by a coherent state. When far apart, the entropy of each is increasing but when getting closer to each other, the entropy is oscillating, demonstrating that the system is in \setOscillating, as shown in Figure~\ref{fig:two-particle-collision-entropy}.

We speculate that some physical phenomena, such as high-speed collision $\ee^+ + \ee^- \rightarrow 2 \gamma$, produce new particles when the entropy is about to decrease, and a transformation must occur according to the  entropy law proposed. Thus the law restricts which outcomes from particle collisions can take place. Of course, a more quantitative analysis is needed to verify that this law predicts those type of physical phenomena.

\section{Conclusions}
\label{sec:conclusion}

We proposed an entropy and an entropy law to govern quantum laws of \physics,  providing the time arrow as the arrow of information loss. Physical phenomena  reported as spontaneous transitions or decays or outcomes of new particles from particle collisions may be caused by the entropy law when systems that are in oscillatory states must  transition to other states to assure an increase entropy while satisfying conservation laws.  Some transition of particles into antiparticles can also be governed by this law.  Other explorations of this law would be to study quantitatively the entropy in beta decay processes and particle collisions.
This law applied to oscillatory behaviors of particles suggests that  instantaneous events of creation and annihilation of particles  occur and are irreversible.

We are left wondering what happens at the instance when particles are transformed into other particles or new particles are created. Could the entropy decrease in those instantaneous events?
After all, the creation of new particles may require an increase in information, and if so, is there a source of information in the vacuum for decreasing the entropy?
If so,  after one of those transformations,  moving backward in time for shorter periods of time to increase entropy can occur during oscillations. Also, could it be that the so called ``collapse of the wave function'' after measurements are made, a much controversial subject with different schools of thought, is also a result of the entropy law and not just a decoherence effect?

While our focus of development here was for fermions, a study of entropy evolution for bosons could also be of interest. Also, the study of retarded potentials, solutions to Maxwell equations or other Gauge fields that travel backwards in time,  may become clearer when studied under the entropy law.

The proposed entropy law governs  a one-particle system and thus it is an intrinsic property of particles.  Quantum \mechanics is thought to become classical \mechanics via one or more of the following transformations: the limit  $ \planckbar \rightarrow 0$,   Ehrenfest theorem, or the WKB approximation or  decoherence, though no complete proof exists. In a system of multiple particles, such approximations may  lead to statistical \mechanics. In that case, the proposed entropy  may lead to the  Gibbs entropy (up to the Boltzmann constant), and quantum effects may account for the blur needed to prove the H-theorem of Gibbs. Then,  the second law of thermodynamics  may follow from the entropy law proposed here.

\section{Acknowledgement}
This material is partially based upon work supported by both the National Science Foundation under Grant No.~DMS-1439786 and the Simons Foundation Institute Grant Award ID 507536 while the first author was in residence at the Institute for Computational and Experimental Research in Mathematics in Providence, RI, during the spring 2019 semester ``Computer Vision''  program.

 \appendix
 \section{Dirac Spinors}
 \label{sec:quantum-review}

We now review the dispersion properties of the \Schroedinger and Dirac equations. The motivation is to show that these two Hamiltonians have intrinsic properties to disperse  any localized initial fermion  distribution,  as we study in Section~\ref{sec:dispersion-coherent-states}.

The free particle Hamiltonians are $H^{\mathscr{S}} ={\frac {-\planckbar ^{2}}{2m }}\nabla ^{2}$ and $H^{\mathscr{D}} = -\iu \planckbar\gamma^0 \vec{\gamma} \cdot \nabla+ m c \gamma^0 $, for the \Schroedinger equation (superscript ${\mathscr{S}}$) and the Dirac equation (superscript ${\mathscr{D}}$), respectively. We refer to the Dirac spinor solution by $\psi(\mathbf{r},t)$.

These are descriptions in position-time space, and we can also write them in Fourier space. Both Hamiltonians are functions of the momentum operator and therefore can be diagonalized in the
spatial Fourier domain $\ket{\mathbf {k}}$ basis,  to obtain respectively
\begin{align}
\omega^{\cal S}(\mathbf{k}) & =\frac{\planckbar}{2m} \matrixsym{k}^2\,,
\\
\omega^{\cal D}(\mathbf{k})& =\pm c \sqrt{ \matrixsym{k}^2+\frac{m^2}{\planckbar^2} c^2}\,,
\label{eq:Fourier-Hamiltonians}
\end{align}
where $\omega^{\cal S,D}(\mathbf{k})$ are the frequency components of the respective Hamiltonians.  The group velocity becomes respectively
 \begin{align}
\mathbf{v_g}^{\cal S}(\mathbf{k})&=\nabla_{\mathrm{k}} \omega^{\cal S}(\mathbf{k})=\frac{\planckbar}{m} \mathbf {k} \,,
\\
\mathbf{v_g}^{\cal D}(\mathbf{k})&=\nabla_{\mathrm{k}} \omega^{\cal D}(\mathbf{k})=\pm \frac{\planckbar}{m}\frac{\mathbf {k}}{\sqrt{1+(\frac{\planckbar \matrixsym{k}}{m c})^2}}\,.
\label{eq:Fourier-group-velocity}
 \end{align}
In  \eqref{eq:w-dispersion} we use the Taylor expansions of~\eqref{eq:Fourier-Hamiltonians} up to the second order, thus  requiring  the Hessians $\hessian_{ij}^{\cal S,D}(\mathbf{k}) =\frac{\partial^2 \omega^{\cal S,D}(\mathbf{k}) }{\partial \matrixsym{k}_i \partial \matrixsym{k}_j}$ with
\begin{align}
\hessian_{ij}^{\cal S}(\mathbf{k}) &=\frac{\planckbar}{m} \deltau_{i,j} \,,
\\
\hessian_{ij}^{\cal D}(\mathbf{k})&= \pm \frac{\planckbar}{m} \left (1+\left (\frac{\planckbar \matrixsym{k}}{m c}\right)^2\right)^{-\frac{3}{2}}\left [ \deltau_{i,j}\left (1+\left (\frac{\planckbar \matrixsym{k}}{m c}\right)^2 \right) -\left (\frac{\planckbar \matrixsym{k}_i}{m c}\right) \left (\frac{\planckbar \matrixsym{k}_j}{m c}\right)\right ]\, .
\label{eq:Fourier-group-Hessian}
\end{align}

A Hessian gives a measure of dispersion of the wave. For the \Schroedinger equation,  the larger is the mass of the particle, the smaller is the dispersion. The eigenvalues for the Dirac equation are
\begin{align}
\label{eigenvalues-Hamiltonian}
\lambda_1^D&=\pm \frac{\planckbar}{m}\left (1+\left (\frac{\planckbar \matrixsym{k}}{m c}\right)^2\right)^{-\frac{3}{2}}=\pm \planckbar \frac{m^2}{\left (m^{2}+\mu^2(\matrixsym{k})\right)^{\frac{3}{2}}}\,,
\\
\lambda_{2,3}^D& =\pm \frac{\planckbar}{m} \left (1+\left (\frac{\planckbar \matrixsym{k}}{m c}\right)^2\right)^{-\frac{1}{2}}=
\pm \planckbar \frac{1}{(m^2+\mu^2(\matrixsym{k}))^{\frac{1}{2}} }\, ,
\end{align}
where $\mu(\matrixsym{k})={\planckbar \matrixsym{k}}/{c}$ is a measure of the kinetic energy in mass units, and the second eigenvalue has multiplicity two. Thus, for both equations the Hessian is positive definite for positive energy solutions. For $\lambda_{2,3}^\mathscr{D}$, the larger is the mass of the particle, the smaller are the eigenvalues and the dispersion. However, for $\lambda_1^\mathscr{D}$ and  $m \le \sqrt{2} \mu(\matrixsym{k})$,
the eigenvalues and dispersion  increase as the mass increases.

\section{CPT  Transformations}
\label{sec:CPT}

In Section~\ref{subsec:cpt} we show the entropy to be invariant under CPT transformations and in  Section~\ref{subsec:time-reflection} we investigate transformation in a QFT that map \QCurves in \setDecreasing to \QCurves in \setIncreasing. A brief review of discrete symmetries that could be present in a QFT may then be useful. More specifically, we  briefly review the  operations of Parity Change, Charge Conjugation, and Time Reversal associated with possibly three discrete symmetries of a QFT. Our focus is on Dirac spinors (bispinors), but  the main concepts  hold for scalar fields and  bosons as well. We start with the motion equation
\begin{align}
    \mathrm{i} \planckbar {\frac {\partial }{\partial t}} \Psi(\mathbf{r},t) =H \Psi(\mathbf{r},t)\,,
\end{align}
and we use Dirac's Hamiltonian $H^{\mathscr{D}} = -\iu \planckbar\gamma^0 \vec{\gamma} \cdot \nabla+ m c \gamma^0 $ to specify the fermions.

The results below apply also to a  QFT  where the Hamiltonian  is given by
\begin{align}
    {\cal H}^{\mathscr{D}}=\int \diff^3 \mathbf{r} \, \Psi^{\dagger}(\mathbf{r},t) \left ( -\iu \planckbar\gamma^0 \vec{\gamma} \cdot \nabla+ m c \gamma^0  \right )\Psi(\mathbf{r},t)\,.
\end{align}

To extend the results for the gauge field $A^{\mu}=(\phi, \vec A)^{\tran}$ we can replace the derivative $\partial_{\mu} \mapsto D_{\mu}=\partial_{\mu} -\iu \frac{e}{c} A_{\mu}$ in the motion equation. Then,  we obtain the well-known properties of those operators acting on photons. We have $PA^{\mu}(\mathbf{r},t)P^{-1}=-A^{\mu}(-\mathbf{r},t)$, $C A^{\mu}(\mathbf{r},t) C^{-1}=-A^{\mu}(\mathbf{r},t)$, $T A^{\mu}(\mathbf{r},t) T^{-1}=-A^{\mu}(\mathbf{r},t)$. Similarly, one can extend the covariant derivative to include all gauge fields and to obtain the properties of these symmetries acting on such fields.

For $\psi^{\mathrm{P}}$, $\psi^{\mathrm{C}}$, and $\psi^{\mathrm{T}}$, used in the following discussion, see Definition~\ref{def:cpt-states},

\subsection{Parity  Change }

The  operator $P$  effects the transformation $\mathbf{r} \mapsto -\mathbf{r}$. The main property $P$ must satisfy  is
\begin{align}
  P H_{-r} P^{-1} &= H  \qquad \text{(invariance constraint to solve for $P$)}\,,
\label{eq:parity-change-P}
\end{align}
where
\begin{align}
H_{-r} &=  \iu \planckbar\gamma^0 \vec{\gamma} \cdot \nabla_{-\mathbf{r}}+ m c \gamma^0\,.
\end{align}
Thus, \eqref{eq:parity-change-P} yields
\begin{align}
    P  \gamma^0 P^{-1} &= \gamma^0 \qquad {\rm and} \qquad P  \gamma^0 \vec{\gamma} P^{-1}=-\gamma^0 \vec{\gamma}
    \\
    P=\gamma^0 & \qquad \text{( unitary solution up to a sign).}
\end{align}
and $\Psi^{\mathrm{P}}(-\mathbf{r},t)$  is also a solution of the motion equation.

\subsection{Charge Conjugation}

Charge conjugation  transforms particles $\Psi(\mathbf{r},t)$ into antiparticles  $\overline{\Psi}^\tran(\mathbf{r},t)= (\Psi^{\dagger}\gamma^0)^\tran(\mathbf{r},t)$. We focus on Dirac spinors where charge conjugation is described by  an operator $\mathscr{C}=C \hat K^\mathrm{C}$, where $\hat K^\mathrm{C}$ acts on an operator $O$ as follows $\hat K^C O=\overline{O}^{\tran}= \left((\gamma^0 O)^{\dagger}\right)^{\tran}$. Then, the main property  $C$ must satisfy is
  \begin{align}
    -C\left( \gamma^0 H \gamma^0\right)^{\tran} C^{-1} &= H\,.
  \label{eq:charge-conjugation-C}
\end{align}
Thus, $ C \gamma^{\mu } C^{-1}= -  \gamma^{\mu \tran}$ and  $\Psi^{\mathrm{C}}(\mathbf{r},t)$ is  also  a solution to the motion equation.  In the standard representation $C=\iu \gamma^2\gamma^0$, up to a phase.

\subsection{Time Reversal}

The transformation effects $t \mapsto -t$ and is carried by the operators $\mathscr{T}=T \hat K$, where $\hat K$ applies conjugation, so $\hat K$ is an antilinear and antiunitary operator, while $T$ is a linear unitary operator which  acts on the spinor structure.  Thus $\mathscr{T}$ is an antilinear and an antiunitary   operator.  Then, the main property  $T$ must satisfy is
\begin{align}
  T H^*T^{-1}  &= H
  \label{eq:time-reversal-T}
\end{align}
and $\Psi^\mathrm{T}(\mathbf{r},-t)$ is  also  a solution to the motion equation. For fermions
\begin{align}
    T  \gamma^{0*} T^{-1} &= \gamma^0 \qquad {\rm and} \qquad T  \vec{\gamma}^* T^{-1}= \vec{\gamma}
    \\
    T=\iu \gamma^1\gamma^3 & \qquad \text{(standard representation, up to a phase)}.
\end{align}
Note that $[H,\mathscr{T}]=0$ and $\mathscr{T}^2= T \hat K T \hat K=T T^* \hat K \hat K=T T^*=-\Identitymatrix$. This implies that Kramers degeneracy applies.

\bibliographystyle{abbrv}
\bibliography{gk01}
\end{document}